\documentclass[journal]{IEEEtran}

\usepackage{amsmath,amssymb,amsfonts,amsthm}
\usepackage{booktabs}
\usepackage{graphicx}
\usepackage{multirow}
\usepackage{array}
\usepackage{tikz}
\usepackage{cite}
\usepackage{xcolor}
\usepackage{url}
\usepackage{algorithm}
\usepackage{algpseudocode}
\usetikzlibrary{arrows.meta,positioning,calc,backgrounds}
\usepackage[hidelinks]{hyperref}

\newtheorem{lemma}{Lemma}

\title{GenAI-FDIA: Physics-Informed Generative Models for False Data Injection Attacks}

\author{Mohammad A. Razzaque,~\IEEEmembership{Senior Member,~IEEE,}
        and~Muta Tah~Hira%
\thanks{Mohammad A. Razzaque (Associate Professor, Research), School of Computing, Engineering and Digital Technologies, Teesside University, UK, E-mail: m.razzaque@tees.ac.uk.}%
\thanks{Muta Tah Hira, Smartifier Ltd, Stockton-on-Tees, UK, E-mail: mhira@smartifier.co.uk.}%
\thanks{Manuscript received April 27, 2026; revised, 2026.}}

\begin{document}
\maketitle

%% -- Abstract ---------------------------------------------------------------
\begin{abstract}
Training and evaluating false data injection attack (FDIA) detectors for power systems is constrained by data scarcity. Operational grid measurements are commercially sensitive, and hand-crafted attacks fail to capture the distributional structures imposed by network physics. We present \textsc{GenAI-FDIA}, a  novel framework benchmarking a pool of $P{=}20$ architectures for physics-compliant FDIA synthesis, spanning Wasserstein GANs (Generative Adversarial Networks), MMD-VAEs (Maximum Mean Discrepancy Variational Autoencoders), normalising flows, diffusion models, and cross-family hybrids. These are evaluated across three IEEE testbeds (14-bus DC, 30-bus DC, and 14-bus AC) under a 60/20/20 chronological split using data-driven Bad Data Detection (BDD) threshold calibration. All architectures achieve evasion rates of $\epsilon_{\text{BDD}} \ge 86.6\%$ on the 14-bus network, and limiting an attacker's topological knowledge induces a measurable degradation in stealthiness ($p \le 0.0022$). We identify a previously unreported failure mode; applying affine physics projections in normalised feature spaces critically displaces the attack vector, collapsing BDD evasion from ${\sim}55\%$ to $<\!2\%$ on the 30-bus testbed. We resolve this via an inference-time harmoniser, restoring full stealthiness ($\epsilon_{\text{BDD}}{=}100\%$) across all physics-informed variants without retraining. We further isolate a covariance-collapse phenomenon ($\kappa \approx -0.076$) within hybrid architectures and rectify it through 50-epoch warm-up schedules ($\kappa \to 0.785$, $\Delta\text{MMD}={-}3.1\%$). \textsc{GenAI-FDIA} delivers a recovery blueprint applicable to any physics-constrained generative model deployed for power-system security.
\end{abstract}

\begin{IEEEkeywords}
False data injection attacks, generative models, variational autoencoders,
diffusion models, physics-informed generation, power-system security,
explainability.
\end{IEEEkeywords}

%% ==========================================================================
\section{Introduction}
\label{sec:introduction}

Modern power grids are simultaneously decarbonising, digitalising, and decentralising. Renewable generation now contributes a significant fraction of electricity supply across major regions \cite{irena2024renewables} and is supported by distributed energy resources, storage, and vehicle-to-grid infrastructure. The resulting cyber-physical interdependencies have widened the attack surface beyond what conventional security systems were designed to monitor. Recent incidents, the BlackEnergy-driven distribution outages in Ukraine in 2015 \cite{lee2016ukrainegrid}, the INDUSTROYER2 deployment against Ukrainian high-voltage substations in 2022 \cite{eset2022industroyer2}, and the Volt Typhoon pre-positioning operations disclosed inside US critical infrastructure \cite{cisa2024volttyphoon}, demonstrate that energy infrastructure is now an explicit, well-resourced target.

State-estimation pipelines are particularly exposed. False Data Injection Attacks (FDIAs) \cite{liu2011fdia} exploit the algebraic structure of the linear state-estimation model by constructing injection vectors $\mathbf{c}\in\operatorname{col}(\mathbf{H})$, where $\mathbf{H}$ is the measurement Jacobian, so that an arbitrary bias is added to the estimated bus-voltage angles while the BDD residual remains unchanged. Undetected, such manipulation can drive suboptimal generation dispatch, suppress protective-relay operation, or mask incipient faults \cite{tian2022afdias}.

Detector quality depends directly on the diversity of attack data on which detectors are trained. Real attacks are rare and concealed; controlled testbed experiments cover only a small fraction of the feasible attack space, and the formal enumeration of all BDD-bypassing vectors scales poorly beyond small bus systems \cite{shahriar2022iattackgen}. Generative AI models address this scarcity by learning the statistical distribution of feasible attacks. Conditional Wasserstein GANs trained on formally stealthy corpora \cite{shahriar2022iattackgen} reach $>$90\% BDD evasion on the 14-bus benchmark. Subsequent work specialised the GAN paradigm with sparse adversarial and ResNet--FCN generators \cite{li2023sparsegan}, least-squares GANs \cite{liu2024lsganfdia}, autoencoder-initialised generators for time-varying topologies \cite{huang2022sparseae}, time-series-coherent GANs \cite{hong2023tsganfdia}, and attention--diffusion hybrids \cite{li2024advdiffusionfdia}. Another set of work embeds AC or DC power-flow constraints directly into the synthesis mechanism, including extrapolative adversarial VAEs \cite{zhao2025piexavae}, physics-constrained GANs \cite{jing2025physconstrainedgan}, and physics-informed denoising diffusion \cite{zhang2024piddpm}. ADMM-based synthesisers jointly evade multi-label locational detectors \cite{tian2026admm}, and high-renewable scenarios expand the feasible attack space \cite{liu2025spatiotemporal}. On the detection side, sequential LSTM-autoencoders \cite{musleh2023lstmfdia,xu2023ddetmtd}, graph-convolutional attention networks \cite{xia2024gcat}, transformer--LSTM hybrids \cite{baul2023xtm}, and adversarial anomaly transformers \cite{liu2025adtrans} now raise the empirical evasion bar.

Despite this progress, three compounding limitations (L) restrict the practical utility of existing pipelines. \textbf{(L1)} Each published framework commits to a single generative architecture, with no mechanism for selecting the architecture best suited to a given topology, dataset size, or measurement representation; the across-family variance has not been characterised on a common protocol. Nawaz \emph{et al.}~\cite{nawaz2025gangridfdia} report the only prior multi-architecture comparison, but only across two GAN variants on a single testbed. \textbf{(L2)} Attacker topology knowledge is treated as a binary condition rather than the continuous spectrum along which real adversaries move \cite{li2020conaml}. \textbf{(L3)} Generated attacks are rarely subjected to attribution analysis that would identify which sensors are preferentially manipulated, information that is operationally actionable for sensor hardening. Furthermore, prior physics-informed approaches apply the projection in normalised feature space without analysing the geometric interaction with feature normalisation, a gap that, as we show below, masks a structural failure mode.

This paper presents \textsc{GenAI-FDIA}, addressing these limitations through four integrated contributions:
\begin{enumerate}
  \item \textbf{Multi-model pool with Pareto selection:} Twenty generative architectures from four families (WGAN, MMD-VAE, normalising flow, diffusion) are evaluated on each testbed under a unified protocol. The pool contains four novel cross-family hybrids: \emph{MMD-VAE-WGAN} (augments a Wasserstein autoencoder with an adversarial discriminator), \emph{LSTM-MMD-VAE-WGAN} (prepends a temporal LSTM encoder), \emph{TC-MMD-VAE} (jointly enforces MMD and total-correlation latent disentanglement), and \emph{RealNVP-TC-MMD-VAE} (appends a normalising flow for exact likelihood). Pareto-front selection over fidelity, BDD evasion, physics consistency, and normalised attack impact identifies the non-dominated subset.
  \item \textbf{Physics-informed projection with formal failure analysis:} Each pool member has a physics-informed (PI-) variant whose generator output passes through a constraint-satisfaction projection. We provide a geometric lemma showing that applying this projection in a standard-scaler space displaces the projected sample outside $\operatorname{col}(\mathbf{H})$ whenever the training-data mean lies outside that subspace, and we provide an inference-time harmoniser that resolves the failure without retraining.
  \item \textbf{Continuous attacker-knowledge conditioning:} A scalar $k\!\in\![0,1]$ governs the fraction of $\mathbf{H}$ rows revealed to the generation process, mapping topology-agnostic injection ($k{=}0$, MITRE ATT\&CK T0817 \cite{MITRE}) through partial-knowledge reconnaissance ($k{=}0.5$, T0842) to fully topology-aware synthesis ($k{=}1$, T0843).
  \item \textbf{Cross-level explainability:} Data-level and model-level attribution are applied to Pareto-optimal models, and their cross-level sensor-ranking consistency is quantified by Cohen's~$\kappa$.
\end{enumerate}

The evaluation is conducted on three IEEE benchmark testbeds: 14-bus DC (2{,}000 samples) \cite{shahriar2022iattackgen}, 30-bus DC (2{,}000 samples), and 14-bus AC with 68-dimensional phasor measurements (35{,}130 samples) \cite{xu2023ddetmtd}. Detectors are BDD (legacy residual test) and LSTM-AE \cite{xu2023ddetmtd}; Isolation Forest \cite{liu2008iforest} provides an additional unsupervised probe on the 14-bus DC testbed. Direct numerical comparison against concurrent diffusion-based FDIA frameworks~\cite{li2024advdiffusionfdia,qin2026false} is precluded because their code is not publicly available; the comparison is made at the design level. The remainder of the paper is organised as follows. Section~\ref{sec:methodology} states the threat model, the conditional generative objective, the physics-projection operator together with its formal failure-mode lemma, and the evaluation protocol. Section~\ref{sec:results} reports results across the three testbeds with discussion inlined per subsection. Section~\ref{sec:conclusion} summarises the findings.

%% ==========================================================================
\section{Methodology}
\label{sec:methodology}

\subsection{State estimation, BDD, and threat model}
\label{ssec:bg_se}

Power-system state estimation recovers bus-voltage angles $\hat{\mathbf{x}}\!\in\!\mathbb{R}^{N-1}$ from measurements $\mathbf{z}\!\in\!\mathbb{R}^{M}$ using the linear DC model $\mathbf{z}=\mathbf{H}\mathbf{x}+\boldsymbol{\varepsilon}$, $\boldsymbol{\varepsilon}\!\sim\!\mathcal{N}(\mathbf{0},\mathbf{R})$, with weighted least-squares solution $\hat{\mathbf{x}}=(\mathbf{H}^{\!\top}\mathbf{R}^{-1}\mathbf{H})^{-1}\mathbf{H}^{\!\top}\mathbf{R}^{-1}\mathbf{z}$ \cite{FDIA2016}. BDD accepts the measurement set when the weighted residual
\begin{equation}
  J(\mathbf{z}) = (\mathbf{z}-\mathbf{H}\hat{\mathbf{x}})^{\!\top}\mathbf{R}^{-1}(\mathbf{z}-\mathbf{H}\hat{\mathbf{x}})
  \label{eq:bdd}
\end{equation}
falls below a $\chi^{2}$ threshold $\tau$. Following \cite{liu2011fdia}, an FDIA $\mathbf{c}=\mathbf{H}\boldsymbol{\delta}\!\in\!\operatorname{col}(\mathbf{H})$ biases the state estimate by $\boldsymbol{\delta}$ while leaving the residual unchanged. Components in $\operatorname{col}(\mathbf{H})^{\perp}=\ker(\mathbf{H}^{\!\top})$ excite non-zero residual energy and are detectable. For 14-bus DC, $\mathbf{H}\!\in\!\mathbb{R}^{54\times 13}$ partitions the measurement space into a 13-dimensional stealthy subspace and a 41-dimensional detectable complement; for 30-bus DC, $\mathbf{H}\!\in\!\mathbb{R}^{71\times 29}$ yields a 29-dimensional stealthy subspace.

The adversary issues an injection $\mathbf{c}\!\in\!\mathbb{R}^{M}$ producing $\mathbf{z}_{a}=\mathbf{z}+\mathbf{c}$. We parameterise the attacker by a scalar $k\!\in\![0,1]$ controlling the fraction of $\mathbf{H}$-rows revealed, assigned to MITRE ATT\&CK ICS sub-techniques as $k{=}0\!\to$T0817 (drive-by compromise), $k{=}0.5\!\to$T0842 (network sniffing), $k{=}1\!\to$T0843 (program download) \cite{MITRE}. The objective is to learn a conditional generator $G_{\theta}\!:\mathcal{Z}\!\times\!\mathbb{R}^{c_{\mathrm{dim}}}\!\to\!\mathbb{R}^{M}$ approximating $p(\mathbf{c}\,|\,\mathbf{H}_{k})$, where $\mathbf{H}_{k}$ encodes the $k$-fraction of $\mathbf{H}$ that the adversary observes.

\subsection{Block-wise framework and chronological partition}
\label{ssec:framework}

\begin{figure}[!t]
\centering
\includegraphics[width=\columnwidth]{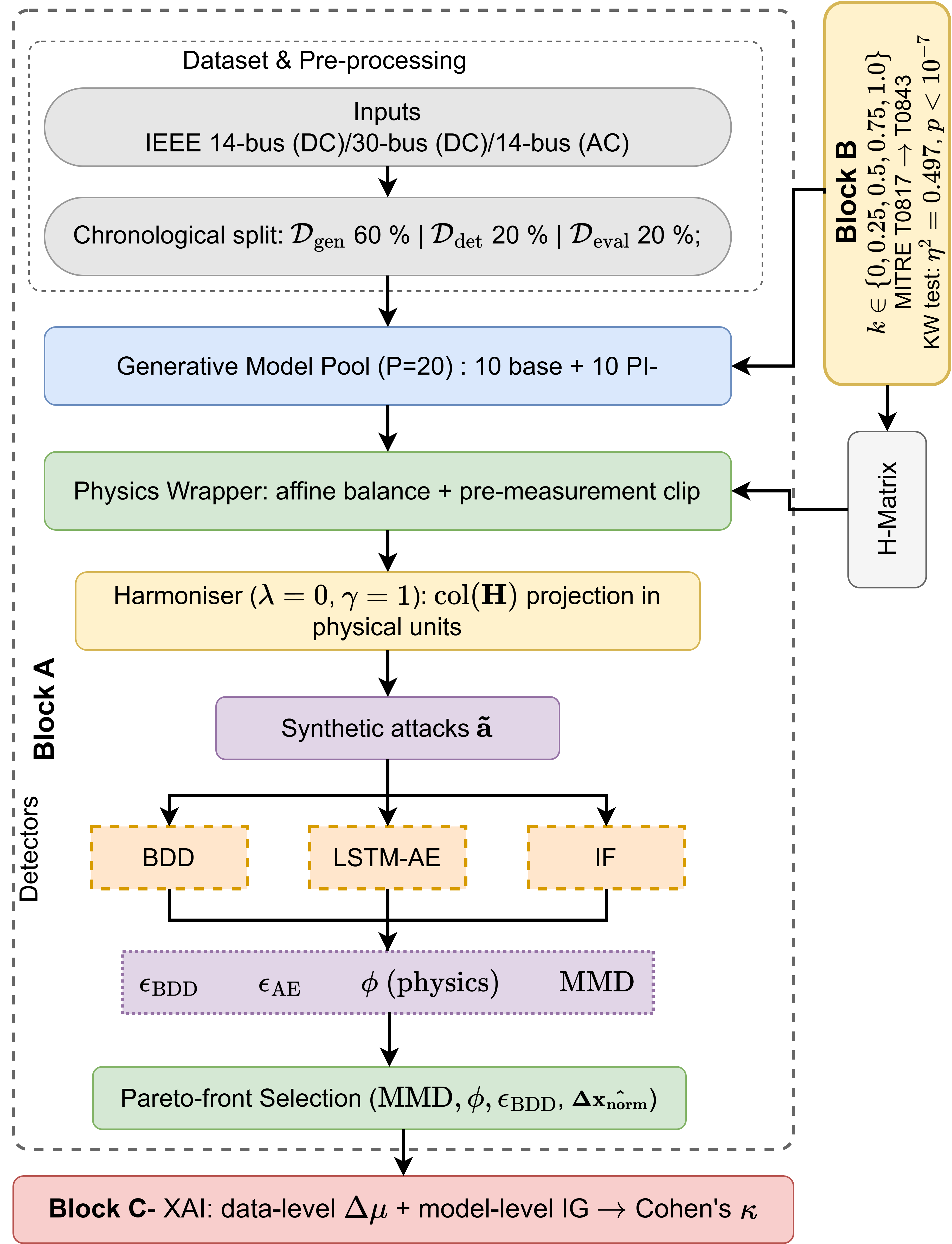}
\caption{GenAI-FDIA evaluation pipeline. The pool ($P{=}20$) trains on $\mathcal{D}_{\mathrm{gen}}$; PI- variants pass outputs through the \textsc{PhysicsWrapper} and inference-time harmoniser. Block~A selects Pareto-optimal models over MMD, $\epsilon_{\mathrm{BDD}}$, $\phi$, and $\Delta\hat{x}_{\mathrm{norm}}$. Block~B re-evaluates a representative per-family triplet over five attacker-knowledge levels $k$. Block~C applies cross-level XAI attribution to Pareto-optimal configurations.}
\label{fig:meth_pipeline}
\end{figure}

The evaluation is divided into three blocks (Fig.~\ref{fig:meth_pipeline}). Block~A benchmarks the entire pool and applies Pareto-front selection. Block~B characterises the effect of attacker knowledge $k$ on a representative per-family triplet (one base, one PI- variant, one cross-family PI- hybrid) drawn from each testbed; this design isolates knowledge-induced variance from family-induced variance, which a strict Pareto-optimal subset (typically dominated by a single family per testbed) cannot. Block~C applies cross-level XAI attribution to Pareto-optimal configurations. To prevent the inflation of evasion rates that occurs when generative training, detector training, and evasion evaluation share data, every dataset is partitioned chronologically (\texttt{shuffle=false}, \texttt{seed=42}) into three non-overlapping segments: 60\% $\mathcal{D}_{\mathrm{gen}}$, 20\% $\mathcal{D}_{\mathrm{det}}$, and 20\% $\mathcal{D}_{\mathrm{eval}}$.

%% ── Figure 1 (near Section B text) ──────────────────────────────────────
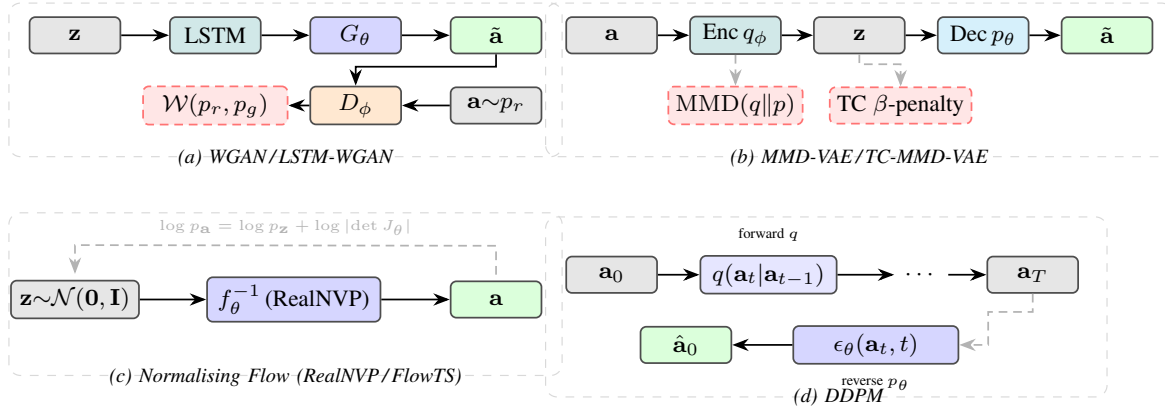
\begin{figure*}[!t]
\centering
\resizebox{0.85\textwidth}{!}{%
\begin{tikzpicture}[
  font=\footnotesize,
  >=Stealth,
  every path/.style={semithick},
  nd/.style={draw=black!70, rectangle, rounded corners=2pt,
             minimum width=1.10cm, minimum height=0.44cm,
             inner sep=2.5pt, align=center},
  nlat/.style={nd, fill=gray!20},
  ngen/.style={nd, fill=blue!16},
  nenc/.style={nd, fill=teal!18},
  ndec/.style={nd, fill=cyan!14},
  ndis/.style={nd, fill=orange!18},
  nout/.style={nd, fill=green!14},
  npen/.style={nd, fill=red!10, draw=red!55, densely dashed},
  farr/.style={->},
  darr/.style={->, densely dashed, gray!60},
]

%% ── Panel (a): WGAN / LSTM-WGAN ──────────────────────────────────────────
\node[nlat]                       (az)    at (0.00, 0.85) {$\mathbf{z}$};
\node[nenc]                       (alstm) at (1.70, 0.85) {LSTM};
\node[ngen]                       (aG)    at (3.40, 0.85) {$G_\theta$};
\node[nout]                       (agen)  at (5.10, 0.85) {$\tilde{\mathbf{a}}$};
\node[npen, minimum width=1.75cm] (aW)    at (1.70, 0.00) {$\mathcal{W}(p_r,p_g)$};
\node[ndis]                       (aD)    at (3.40, 0.00) {$D_\phi$};
\node[nlat]                       (areal) at (5.10, 0.00) {$\mathbf{a}{\sim}p_r$};

\draw[farr] (az)         -- (alstm);
\draw[farr] (alstm)      -- (aG);
\draw[farr] (aG)         -- (agen);
\draw[farr] (agen.south) -- ++(0,-0.15) -| (aD.north);
\draw[farr] (areal)      -- (aD);
\draw[farr] (aD)         -- (aW);

\node[font=\scriptsize\itshape, anchor=north] at (2.55,-0.40)
  {(a)~WGAN\,/\,LSTM-WGAN};
\begin{pgfonlayer}{background}
  \draw[draw=gray!30, dashed, rounded corners=4pt, line width=0.4pt]
    (-0.80,-0.62) rectangle (5.90, 1.24);
\end{pgfonlayer}

%% ── Panel (b): MMD-VAE / TC-MMD-VAE  [xshift = 8.5 cm] ──────────────────
\begin{scope}[xshift=6.5cm]
  \node[nlat]                       (ba)   at (0.00, 0.85) {$\mathbf{a}$};
  \node[nenc]                       (benc) at (1.50, 0.85) {Enc\,$q_\phi$};
  \node[nlat]                       (bz)   at (3.00, 0.85) {$\mathbf{z}$};
  \node[ndec]                       (bdec) at (4.50, 0.85) {Dec\,$p_\theta$};
  \node[nout]                       (bgen) at (6.00, 0.85) {$\tilde{\mathbf{a}}$};
  \node[npen, minimum width=1.55cm] (bmmd) at (1.50, 0.00) {$\operatorname{MMD}(q\Vert p)$};
  \node[npen, minimum width=1.55cm] (btc)  at (3.50, 0.00) {TC\;$\beta$-penalty};

  \draw[farr] (ba)         -- (benc);
  \draw[farr] (benc)       -- (bz);
  \draw[farr] (bz)         -- (bdec);
  \draw[farr] (bdec)       -- (bgen);
  \draw[darr] (benc.south) -- (bmmd.north);
  \draw[darr] (bz.south)   -- ++(0,-0.12) -| (btc.north);

  \node[font=\scriptsize\itshape, anchor=north] at (3.00,-0.40)
    {(b)~MMD-VAE\,/\,TC-MMD-VAE};
  \begin{pgfonlayer}{background}
    \draw[draw=gray!30, dashed, rounded corners=4pt, line width=0.4pt]
      (-0.80,-0.62) rectangle (6.80, 1.24);
  \end{pgfonlayer}
\end{scope}
%% ── Panel (c): Normalising Flow (RealNVP / FlowTS)  [yshift = -2.95 cm] ───
\begin{scope}[yshift=-2.95cm]
  \node[nlat, minimum width=1.55cm] (cz) at (0.00, 0.60)
    {$\mathbf{z}{\sim}\mathcal{N}(\mathbf{0},\mathbf{I})$};
  \node[ngen, minimum width=2.10cm] (cf) at (2.65, 0.60)
    {$f_\theta^{-1}$\,(RealNVP)};
  \node[nout]                       (ca) at (5.10, 0.60) {$\mathbf{a}$};

  \draw[farr] (cz) -- (cf);
  \draw[farr] (cf) -- (ca);
  %% dashed likelihood relation: vertical → horizontal → vertical
  \path (ca.north) ++(0,0.42) coordinate (ctop);
  \path (cz.north) ++(0,0.42) coordinate (cztop);
  \draw[darr, shorten <=2pt, shorten >=2pt]
    (ca.north) -- (ctop)
    -- node[above, font=\tiny]
         {$\log p_{\mathbf{a}} = \log p_{\mathbf{z}} + \log|\!\det J_\theta|$}
       (cztop)
    -- (cz.north);

  \node[font=\scriptsize\itshape, anchor=north] at (2.55,-0.10)
    {(c)~Normalising Flow (RealNVP\,/\,FlowTS)};
  \begin{pgfonlayer}{background}
    \draw[draw=gray!30, dashed, rounded corners=4pt, line width=0.4pt]
      (-0.80,-0.32) rectangle (5.90, 1.65);
  \end{pgfonlayer}
\end{scope}
%% ── Panel (d): DDPM  [xshift = 8.5 cm, yshift = -2.95 cm] ───────────────
\begin{scope}[xshift=6.5cm, yshift=-2.95cm]
  \node[nlat]                       (da0f)  at (0.00, 0.90) {$\mathbf{a}_0$};
  \node[nd, fill=blue!10, minimum width=1.65cm] (dq) at (1.90, 0.90)
    {$q(\mathbf{a}_t|\mathbf{a}_{t-1})$};

  \node[font=\footnotesize]         (ddots) at (3.75, 0.90) {$\cdots$};
  \node[nlat]                       (daT)   at (5.10, 0.90) {$\mathbf{a}_T$};
  \node[ngen, minimum width=2.00cm] (deps)  at (3.20, 0.05)
    {$\epsilon_\theta(\mathbf{a}_t,t)$};
  \node[nout]                       (da0r)  at (0.90, 0.05) {$\hat{\mathbf{a}}_0$};

  \draw[farr] (da0f)  -- (dq);
  \draw[farr] (dq)    -- (ddots);
  \draw[farr] (ddots) -- (daT);
  %% dashed dependency: run horizontally on a lower rail and enter $\epsilon_\theta$ horizontally
  \path (daT.south) ++(0,-0.22) coordinate (dbar);
  \path (deps.east) ++(0.35,0) coordinate (depsin);
  \path (dbar -| depsin) coordinate (depsbar);
  \draw[darr] (daT.south) -- (dbar) -- (depsbar) -- (depsin) -- (deps.east);
  \draw[farr] (deps)  -- (da0r);

  \node[font=\tiny, above=0.04cm of dq]   {forward\;$q$};
  \node[font=\tiny, below=0.04cm of deps] {reverse\;$p_\theta$};

  \node[font=\scriptsize\itshape, anchor=north] at (2.75,-0.35) {(d)~DDPM};
  \begin{pgfonlayer}{background}
    \draw[draw=gray!30, dashed, rounded corners=4pt, line width=0.4pt]
      (-0.80,-0.58) rectangle (6.00, 1.60);
  \end{pgfonlayer}
\end{scope}

\end{tikzpicture}%
}% end resizebox

\caption{%
  Generative model families in the 20-member pool (base variants shown;
  PI- variants add a \textsc{PhysicsWrapper} at the output).
  \textbf{(a)}~WGAN\,/\,LSTM-WGAN: latent noise passes through an optional
  LSTM encoder before generator $G_\theta$; discriminator $D_\phi$ minimises
  the Wasserstein-1 distance $\mathcal{W}(p_r,p_g)$.
  \textbf{(b)}~MMD-VAE\,/\,TC-MMD-VAE: encoder $q_\phi$ is regularised by an
  MMD penalty; the TC variant adds a $\beta$-weighted total-correlation term
  (dashed) to encourage latent factorisation.
  \textbf{(c)}~Normalising flow (RealNVP\,/\,FlowTS): generation uses the
  inverse map $f_\theta^{-1}$; training (dashed arc) optimises exact
  log-likelihood via the change-of-variables formula.
  \textbf{(d)}~DDPM: a forward diffusion process $q$ corrupts
  $\mathbf{a}_0$ to $\mathbf{a}_T\!{\sim}\mathcal{N}(\mathbf{0},\mathbf{I})$;
  denoising network $\epsilon_\theta$ learns the reverse process $p_\theta$.%
}
\label{fig:family_arch}
\end{figure*}

\subsection{Generative pool}
\label{ssec:pool}

The pool comprises $P{=}20$ models: ten base architectures and their PI- counterparts. The bases span four families, and Figure~\ref{fig:family_arch} displays the data-flow diagrams of these four
generative families. \emph{WGAN/LSTM-WGAN}: a generator $G_\theta\!:\mathcal{Z}\!\to\!\mathbb{R}^{M}$ maps latent noise to synthetic samples; the discriminator minimises the Earth-Mover distance with a gradient penalty \cite{gulrajani2017wgangp}; LSTM-bearing variants prepend a stacked LSTM encoder. \emph{MMD-VAE}: the WAE-MMD variant \cite{tolstikhin2018wae} replaces the KL divergence with $\operatorname{MMD}(q(\mathbf{z}),p(\mathbf{z}))$; the total-correlation variant adds $\beta\,\mathrm{KL}\!\left(q(\mathbf{z})\,\|\,\prod_i q(z_i)\right)$ to encourage latent factorisation \cite{chen2018tcvae}, which, when applied from epoch zero, severs the negative latent correlations $\mathbf{c}=\mathbf{H}\boldsymbol{\delta}$ requires, producing the covariance collapse mode resolved in Section~\ref{ssec:beta_warmup}. \emph{Normalising flows}: invertible transformations map a Gaussian base to the attack distribution via $\mathbf{a}=f_\theta(\mathbf{z})$ with exact log-likelihood \cite{dinh2017realnvp}; FlowTS uses a rectified-flow formulation for time series \cite{hu2025flowts}. \emph{DDPM}: a forward process adds Gaussian noise in $T$ steps, and a network $\epsilon_\theta$ is trained to reverse it \cite{ho2020ddpm}. \emph{MMD-VAE-WGAN}, \emph{LSTM-MMD-VAE-WGAN}, \emph{TC-MMD-VAE}, and \emph{RealNVP-TC-MMD-VAE} are novel cross-family hybrids. The MMD encoder regulariser was chosen over KL because the latter is prone to uninformative latent codes that undermine physical plausibility \cite{tolstikhin2018wae}. Per-family hyperparameters use Adam with $\beta_1{=}0.5,\beta_2{=}0.999$, learning rates in $\{1{-}5\}{\times}10^{-4}$, batches 64--256, and 150--300 epochs depending on the family.

\subsection{Physics-constraint projection and failure-mode lemma}
\label{ssec:physics}

Each PI- model passes its raw output $\tilde{\mathbf{c}}$ through a two-stage \textsc{PhysicsWrapper}. Stage one enforces the DC power-balance constraint $\mathbf{1}^{\!\top}\mathbf{c}=P_{\mathrm{loss}}$ via the affine projection $\mathbf{c}'=\tilde{\mathbf{c}}-\frac{1}{M}(\mathbf{1}^{\!\top}\tilde{\mathbf{c}}-P_{\mathrm{loss}})\mathbf{1}$. Stage two clips each component $c'_{i}$ to the per-measurement voltage window $[\underline{V}_{i}\!-\!z_{i},\overline{V}_{i}\!-\!z_{i}]$. The fraction of samples satisfying both constraints after projection defines the physics-consistency metric $\phi(\mathbf{c})$. A failure mode arises when the same projection is applied in normalised feature space (after \textsc{StandardScaler}); inverse-scaling adds the training-data mean $\boldsymbol{\mu}$ to the projected vector.

\begin{lemma}[Mean leakage out of $\operatorname{col}(\mathbf{H})$]
\label{lem:mean_leak}
Let $\mathbf{P}_{\!\mathbf{H}}$ denote the orthogonal projector onto $\operatorname{col}(\mathbf{H})$, $\boldsymbol{\sigma}\!>\!\mathbf{0}$ the per-feature standard deviation, and $\mathcal{S}(\mathbf{a})=\boldsymbol{\sigma}\!\odot\!\mathbf{a}+\boldsymbol{\mu}$ the inverse \textsc{StandardScaler}. Suppose $\tilde{\mathbf{a}}\!\in\!\operatorname{col}(\mathbf{H})$ in normalised space. Then $\mathbf{a}_{p}=\mathcal{S}(\tilde{\mathbf{a}})$ satisfies
\begin{equation}
  (\mathbf{I}-\mathbf{P}_{\!\mathbf{H}})\,\mathbf{a}_{p} = (\mathbf{I}-\mathbf{P}_{\!\mathbf{H}})\,\boldsymbol{\mu} + (\mathbf{I}-\mathbf{P}_{\!\mathbf{H}})\,(\boldsymbol{\sigma}\!\odot\!\tilde{\mathbf{a}}),
  \label{eq:mean_leak}
\end{equation}
which is non-zero whenever $\boldsymbol{\mu}\!\notin\!\operatorname{col}(\mathbf{H})$ or $\boldsymbol{\sigma}\!\odot\!\tilde{\mathbf{a}}\!\notin\!\operatorname{col}(\mathbf{H})$. Then $J(\mathbf{a}_{p})=\mathbf{a}_{p}^{\!\top}\mathbf{R}^{-1}(\mathbf{I}-\mathbf{P}_{\!\mathbf{H}})\mathbf{a}_{p}>0$ in expectation, so the projected attack is no longer stealthy.
\end{lemma}

\textit{Proof sketch:} By assumption $\mathbf{P}_{\!\mathbf{H}}\tilde{\mathbf{a}}=\tilde{\mathbf{a}}$, so $(\mathbf{I}-\mathbf{P}_{\!\mathbf{H}})\tilde{\mathbf{a}}=\mathbf{0}$. Applying $\mathbf{I}-\mathbf{P}_{\!\mathbf{H}}$ to $\mathbf{a}_{p}=\mathcal{S}(\tilde{\mathbf{a}})$ yields \eqref{eq:mean_leak}. The first term vanishes only when elementwise scaling by $\boldsymbol{\sigma}$ commutes with $\mathbf{P}_{\!\mathbf{H}}$, which requires $\boldsymbol{\sigma}=\sigma\mathbf{1}$; the second term vanishes only when $\boldsymbol{\mu}\!\in\!\operatorname{col}(\mathbf{H})$, equivalent to a flat-start reference angle for every snapshot. Neither condition holds generically. $\blacksquare$

In IEEE 30-bus DC, $\|(\mathbf{I}-\mathbf{P}_{\!\mathbf{H}})\boldsymbol{\mu}\|_{2}\!=\!2.81{\times}10^{2}$ and $\mathbb{E}\,J(\mathbf{a}_{p})\!\approx\!1.4{\times}10^{6}$ versus the calibrated threshold $\tau\!=\!3.30{\times}10^{5}$, a ${\sim}4.2\times$ residual inflation that collapses $\epsilon_{\mathrm{BDD}}$ from ${\sim}55\%$ to $<\!2\%$.

\textbf{Inference-time harmoniser:} The remedy is a single-pass projection in physical units: $\hat{\boldsymbol{\delta}}=(\mathbf{H}^{\!\top}\mathbf{H})^{-1}\mathbf{H}^{\!\top}\mathbf{a}_{p}$, $\mathbf{a}_{h}=\mathbf{H}\hat{\boldsymbol{\delta}}$, with no sanitiser threshold and no iterative refinement (Algorithm~\ref{alg:harmoniser}). The harmoniser restores $\epsilon_{\mathrm{BDD}}{=}100\%$ for all ten PI-variants without retraining.

\begin{algorithm}[!t]
\caption{Inference-time stealth harmoniser}
\label{alg:harmoniser}
\begin{algorithmic}[1]
\Require Generated attack $\mathbf{a}_{p}\!\in\!\mathbb{R}^{M}$ in physical units; Jacobian $\mathbf{H}\!\in\!\mathbb{R}^{M\times(N-1)}$
\Ensure Harmonised attack $\mathbf{a}_{h}\!\in\!\operatorname{col}(\mathbf{H})$
\State $\hat{\boldsymbol{\delta}}\gets(\mathbf{H}^{\!\top}\mathbf{H})^{-1}\mathbf{H}^{\!\top}\mathbf{a}_{p}$ \Comment{Least-squares pseudoinverse}
\State $\mathbf{a}_{h}\gets\mathbf{H}\hat{\boldsymbol{\delta}}$ \Comment{Reproject onto $\operatorname{col}(\mathbf{H})$}
\State \Return $\mathbf{a}_{h}$
\end{algorithmic}
\end{algorithm}

\subsection{Conditioning, Pareto selection, and warm-up}
\label{ssec:conditioning}

\textbf{Continuous conditioning:} The conditioning vector $\mathbf{H}_{k}\!\in\!\mathbb{R}^{c_{\mathrm{dim}}}$ is constructed by flattening the first $\lfloor kM\rfloor$ rows of $\mathbf{H}$ and zero-padding to $c_{\mathrm{dim}}{=}M$. Block~B evaluates $k\!\in\!\{0,0.25,0.5,0.75,1\}$. Conditioning is injected by direct concatenation: the encoder and discriminator receive $[\mathbf{a};\mathbf{H}_{k}]$, the generator receives $[\mathbf{z};\mathbf{H}_{k}]$, and the DDPM denoising network receives $\mathbf{H}_{k}$ as an additional channel at every reverse-diffusion step.

\textbf{Pareto selection:} Models are ranked over four objectives: minimise MMD, maximise $\epsilon_{\mathrm{BDD}}$, maximise $\phi$, and maximise $\Delta\hat{x}_{\mathrm{norm}}$ (normalised attack impact). The fourth axis protects against magnitude collapse: a generator with $\Delta\hat{x}_{\mathrm{norm}}\!\ll\!1$ is statistically stealthy but operationally inconsequential.

\subsubsection{TC \texorpdfstring{$\beta$}{beta}-warm-up schedule}
\label{ssec:beta_warmup}

For RealNVP-TC-MMD-VAE the total-correlation weight $\beta$ is linearly annealed from $0$ to its target over the first $50$ epochs. Applying full $\beta$ from epoch zero produces the covariance-collapse mode reported in Section~\ref{ssec:results_30bus}; the warm-up restores $\kappa$ to $0.785$ at no measurable cost in MMD.

\subsection{Evaluation metrics, detectors, and inference}
\label{ssec:metrics}

Distributional fidelity is measured by Maximum Mean Discrepancy (MMD) with an RBF kernel ($\sigma$ set to the median pairwise distance):
\begin{equation}
\begin{aligned}
  \mathrm{MMD}^{2}(\mathcal{P},\mathcal{Q}) &= \mathbb{E}_{p,p'}[k_{\sigma}(\mathbf{x},\mathbf{x}')] - 2\mathbb{E}_{p,q}[k_{\sigma}(\mathbf{x},\mathbf{y})] \\
  &\quad + \mathbb{E}_{q,q'}[k_{\sigma}(\mathbf{y},\mathbf{y}')].
\end{aligned}
\label{eq:mmd}
\end{equation}
A first-order sliced Wasserstein $W_1$ is reported on 14-bus AC where the unbiased MMD estimator is unstable. Operational severity is $\Delta\hat{x}_{\mathrm{norm}} = \|\hat{\mathbf{c}}\|_{2}/\mathbb{E}_{\mathbf{c}_{r}\in\mathcal{D}_{\mathrm{eval}}}[\|\mathbf{c}_{r}\|_{2}]$; values $\approx 1$ indicate that synthetic attacks match the operational severity of real injections. The metrics table is consolidated in Table~\ref{tab:metric_summary}. All metrics are averaged over $N_{\mathrm{seed}}{=}5$ seeds and reported as mean~$\pm$~SD.

\begin{table}[!t]
\centering
\caption{Evaluation metrics (rows grouped by objective).}
\label{tab:metric_summary}
\begingroup
\fontsize{7}{8.4}\selectfont
\setlength{\tabcolsep}{3pt}
\begin{tabular}{p{1.9cm} p{1.4cm} p{1cm} p{2.6cm}}
\toprule
\textbf{Metric} & \textbf{Symbol} & \textbf{Direction} & \textbf{Role} \\
\midrule
MMD (RBF)              & MMD                       & $\downarrow$ & Pareto axis \\
Wasserstein-1 (sliced) & $W_{1}$                   & $\downarrow$ & AC proxy for MMD \\
BDD evasion            & $\epsilon_{\mathrm{BDD}}$ & $\uparrow$   & Pareto axis \\
LSTM-AE evasion        & $\epsilon_{\mathrm{AE}}$  & $\uparrow$   & Transparency metric \\
Isolation Forest evasion & $\epsilon_{\mathrm{IF}}$ & $\uparrow$  & 14-bus DC probe \\
Physics consistency    & $\phi$                    & $\uparrow$   & Pareto axis \\
Attack impact          & $\Delta\hat{x}_{\mathrm{norm}}$ & $\uparrow$ & Pareto (anti-collapse) \\
Discriminator score    & $\psi$                    & $\to 0.5$    & Calibration check \\
Cross-level agreement  & $\kappa$                  & $\uparrow$   & XAI consistency \\
\bottomrule
\end{tabular}
\endgroup
\end{table}

\textbf{Detectors:} BDD~\cite{liu2011fdia} thresholds $J(\mathbf{z}_{a})$ at $\tau_{\chi^{2}}$. LSTM-AE~\cite{xu2023ddetmtd} flags samples for which the reconstruction error exceeds a threshold $\tau_{\mathrm{ae}}$ learnt on $\mathcal{D}_{\mathrm{det}}$. Isolation Forest~\cite{liu2008iforest} provides a tree-partitioning probe on a 14-bus DC. Both data-driven thresholds are set at the 95th-percentile of the calibration partition and frozen before any attack data are observed.

\textbf{Statistical inference:} Friedman tests pool-level differences on the $20\!\times\!25$ MMD matrix, followed by Nemenyi post-hoc with Bonferroni correction at $\alpha_{\mathrm{adj}}{=}0.0026$. Paired Wilcoxon signed-rank tests evaluate PI/base differences with Cohen's $d$. Kruskal--Wallis tests $k$-level effects with $\eta^{2}=(H-G+1)/(N-G)$. Cohen's $\kappa$ on top-5 sensor rankings agrees with data-level standardised mean shifts and model-level Integrated Gradients on the discriminator. All $p$-values are two-sided at $\alpha{=}0.05$.

\textbf{Explainability:} Data-level attribution measures the per-feature mean shift $\Delta\boldsymbol{\mu}$ between generated and clean baselines. Model-level attribution uses Integrated Gradients (IG) \cite{sundararajan2017ig}: $\phi_j = (a_j - a'_j)\int_{0}^{1}\frac{\partial F(\mathbf{a}'+\alpha(\mathbf{a}-\mathbf{a}'))}{\partial a_j}\,\mathrm{d}\alpha$. Cross-level agreement is Cohen's $\kappa = (p_o - p_e)/(1 - p_e)$ \cite{mchugh2012interrater}, with $\kappa\!=\!1$ representing perfect agreement, $\kappa\!=\!0$ indicating chance agreement, and $\kappa\!<\!0$ reflecting systematic disagreement.

%% ==========================================================================
\section{Experimental Results}
\label{sec:results}

The results are presented in three blocks, following the structure of Section~\ref{ssec:framework}, with a discussion in each block. All values are mean~$\pm$~SD over 25 evaluation runs (5 seeds $\times$, 5 evaluation passes per seed). We conducted a comprehensive set of experiments across the three testbeds, the 20-model pool, the five attacker-knowledge levels, and the four ablation conditions; due to page limits, we present here the subset most relevant to the four contributions, with additional sensor-attribution grids and per-family ablation breakdowns available on request.

\subsection{Block A: 14-bus DC}
\label{ssec:results_14bus}

Table~\ref{tab:14bus} reports all 20 pool members on a 14-bus DC. The Pareto-optimal subset is \emph{PI-DDPM} (rank~1), \emph{LSTM-WGAN}, and \emph{LSTM-MMD-VAE-WGAN}. PI-DDPM achieves the lowest MMD ($0.0174\!\pm\!0.0021$) with $\epsilon_{\mathrm{BDD}}{=}0.940\!\pm\!0.019$ and $\Delta\hat{x}_{\mathrm{norm}}{=}0.97$, indicating that synthetic attacks match the operational severity of real injections to within $3\%$. Across all 20 members, $\epsilon_{\mathrm{BDD}}$ never falls below $0.866$, confirming that the residual test is structurally defeatable by every generative family. LSTM-AE evasion is negligible for the majority of models, with the base DDPM being the sole exception ($\epsilon_{\mathrm{AE}}{=}0.074\!\pm\!0.043$). Physics-informed wrapping consistently elevates $\phi$: PI-LSTM-WGAN reaches $\phi{=}0.637\!\pm\!0.006$ versus $0.246\!\pm\!0.012$ for its base ($2.6\times$). The gain is not universal: diffusion-based PI-variants show lower $\phi$ than their WGAN- and VAE-based peers under the same wrapper, indicating that DDPM dynamics are less amenable to a post-hoc affine power-balance correction.

\begin{table}[!t]
\centering
\caption{Block A on 14-bus DC, sorted by MMD. Pareto-optimal models in bold. SD shown as subscripts where $\ge 0.005$.}
\label{tab:14bus}
\begingroup
\fontsize{6.08}{7.29}\selectfont
\setlength{\tabcolsep}{1.5pt}
\begin{tabular}{l*{6}{>{\fontsize{5.80}{7.00}\selectfont\arraybackslash}r}}
\toprule
\textbf{Model} & \textbf{MMD} & $\boldsymbol{\epsilon}_{\mathrm{BDD}}$ & $\boldsymbol{\epsilon}_{\mathrm{AE}}$ & $\boldsymbol{\epsilon}_{\mathrm{IF}}$ & $\boldsymbol{\phi}$ & $\boldsymbol{\psi}$ \\
\midrule
\textbf{PI-DDPM}        & 0.017 & \textbf{0.940}\,{\tiny$\pm$.019} & 0.000 & 0.063\,{\tiny$\pm$.020} & 0.492\,{\tiny$\pm$.037} & 0.651\,{\tiny$\pm$.042} \\
\textbf{LSTM-WGAN}      & 0.020 & \textbf{1.000} & 0.000 & 0.000 & 0.246\,{\tiny$\pm$.012} & 0.608\,{\tiny$\pm$.048} \\
\textbf{LSTM-MMD-VAE-WGAN} & 0.023 & \textbf{1.000} & 0.000 & 0.004 & 0.235\,{\tiny$\pm$.009} & 0.620\,{\tiny$\pm$.052} \\
PI-FlowTS               & 0.024 & 1.000 & 0.000 & 0.002 & 0.625\,{\tiny$\pm$.007} & 0.683\,{\tiny$\pm$.053} \\
FlowTS                  & 0.025 & 0.998 & 0.000 & 0.000 & 0.235\,{\tiny$\pm$.005} & 0.541\,{\tiny$\pm$.040} \\
DDPM                    & 0.026 & 0.866\,{\tiny$\pm$.031} & 0.074\,{\tiny$\pm$.043} & 0.159\,{\tiny$\pm$.054} & 0.231\,{\tiny$\pm$.008} & 0.549\,{\tiny$\pm$.045} \\
MMD-VAE-WGAN            & 0.035 & 1.000 & 0.000 & 0.000 & 0.298\,{\tiny$\pm$.015} & 0.710\,{\tiny$\pm$.036} \\
PI-LSTM-WGAN            & 0.038 & 1.000 & 0.000 & 0.000 & 0.637\,{\tiny$\pm$.006} & 0.709\,{\tiny$\pm$.047} \\
PI-LSTM-MMD-VAE-WGAN    & 0.043 & 1.000 & 0.000 & 0.001 & 0.625\,{\tiny$\pm$.011} & 0.715\,{\tiny$\pm$.042} \\
PI-MMD-VAE-WGAN         & 0.046 & 1.000 & 0.000 & 0.000 & 0.633\,{\tiny$\pm$.009} & 0.761\,{\tiny$\pm$.036} \\
LSTM-MMD-VAE            & 0.097 & 1.000 & 0.000 & 0.002 & 0.266\,{\tiny$\pm$.008} & 0.641\,{\tiny$\pm$.048} \\
WGAN                    & 0.124 & 1.000 & 0.000 & 0.000 & 0.283\,{\tiny$\pm$.032} & 0.809\,{\tiny$\pm$.019} \\
PI-LSTM-MMD-VAE         & 0.159 & 1.000 & 0.000 & 0.009 & 0.540\,{\tiny$\pm$.012} & 0.767\,{\tiny$\pm$.036} \\
PI-WGAN                 & 0.172 & 1.000 & 0.000 & 0.002 & 0.580\,{\tiny$\pm$.025} & 0.813\,{\tiny$\pm$.023} \\
MMD-VAE                 & 0.236 & 1.000 & 0.000 & 0.038\,{\tiny$\pm$.019} & 0.272\,{\tiny$\pm$.029} & 0.814\,{\tiny$\pm$.019} \\
PI-MMD-VAE              & 0.246 & 1.000 & 0.000 & 0.053\,{\tiny$\pm$.022} & 0.395\,{\tiny$\pm$.028} & 0.813\,{\tiny$\pm$.021} \\
TC-MMD-VAE              & 0.264 & 1.000 & 0.000 & 0.396\,{\tiny$\pm$.032} & 0.216\,{\tiny$\pm$.008} & 0.815\,{\tiny$\pm$.009} \\
RealNVP-TC-MMD-VAE      & 0.264 & 1.000 & 0.000 & 0.390\,{\tiny$\pm$.037} & 0.205\,{\tiny$\pm$.012} & 0.818\,{\tiny$\pm$.015} \\
PI-TC-MMD-VAE           & 0.265 & 1.000 & 0.000 & 0.485\,{\tiny$\pm$.030} & 0.239\,{\tiny$\pm$.008} & 0.828\,{\tiny$\pm$.006} \\
PI-RealNVP-TC-MMD-VAE   & 0.266 & 1.000 & 0.000 & 0.486\,{\tiny$\pm$.029} & 0.227\,{\tiny$\pm$.014} & 0.831\,{\tiny$\pm$.003} \\
\bottomrule
\end{tabular}
\endgroup
\end{table}

\subsection{Block A: 30-bus DC and physics-projection failure}
\label{ssec:results_30bus}

On the 30-bus, FlowTS achieves the lowest MMD ($0.0015\!\pm\!0.0019$) with $\epsilon_{\mathrm{BDD}}{=}0.972\!\pm\!0.016$, making it the dominant Pareto point under the broken \textsc{PhysicsWrapper}; the base DDPM follows ($\mathrm{MMD}{=}0.0113$, $\epsilon_{\mathrm{BDD}}{=}0.912$). The headline finding on this testbed is the physics-projection failure mode of Lemma~\ref{lem:mean_leak}. Eight of the ten PI variants collapse to $\epsilon_{\mathrm{BDD}}\!<\!30\%$ (mean $2.4\%$); the two TC-based PI variants remain at $100\%$ because their latent TC structure incidentally generates samples close to $\operatorname{col}(\mathbf{H})$ in physical units. After the harmoniser is applied (Table~\ref{tab:recovery_p1}), all ten PI- variants reach $\epsilon_{\mathrm{BDD}}{=}1.000$ and the Pareto frontier is reshaped, PI-FlowTS (MMD $0.0016$) and PI-LSTM-WGAN (MMD $0.0032$) join the non-dominated set alongside FlowTS. The collapse is absent in the 14-bus DC because that testbed's calibrated $\tau$ is more permissive relative to residual inflation. LSTM-AE evasion is zero for all 20 models, indicating that the 30-bus LSTM-AE has learnt temporal patterns that no model evades without retraining.

The 14-bus DC testbed masks the failure because its calibrated $\tau$ is permissive enough that imperfect projections remain below it, and LSTM-bearing architectures benefit from temporal modelling that independently improves evasion. The collapse is structural, not incidental: any physics wrapper that separates normalisation from projection will exhibit the same failure whenever the training-data mean lies outside the physics-feasible subspace, the generic case for DC networks operating away from a flat-start reference. The implication is that any physics-informed generative model applying a linear constraint projection after a learnt normalisation must verify that the normalisation mean lies within the constraint subspace or apply the projection in physical units.

\textbf{Inference-time harmoniser:} Applying the harmoniser of Section~\ref{ssec:physics} ($\lambda{=}0$, single-pass projection in physical units) restores $\epsilon_{\mathrm{BDD}}{=}100\%$ for all ten PI variants on 30-bus without retraining; the average WLS $\chi^{2}$ drops from ${\sim}1.4{\times}10^{6}$ to $126{,}862$, a factor of $\sim 2.6$ below the calibrated threshold. Per-model recovery is in Table~\ref{tab:recovery_p1}.

\begin{table}[!t]
\centering
\caption{BDD-evasion recovery via the harmoniser (single pass) on 30-bus DC, $N_{\mathrm{seed}}{=}5$. ``Broken'' uses \textsc{PhysicsWrapper} in normalised space; ``Harm.'' applies projection in physical units.}
\label{tab:recovery_p1}
\begingroup
\fontsize{6.5}{7.8}\selectfont
\setlength{\tabcolsep}{2pt}
\begin{tabular}{l*{4}{>{\fontsize{6.2}{7.5}\selectfont\arraybackslash}r}}
\toprule
Model & $\epsilon_{\mathrm{BDD}}$ (broken) & $\epsilon_{\mathrm{BDD}}$ (harm.) & $\Delta\epsilon_{\mathrm{BDD}}$ & MMD\textsubscript{harm} \\
\midrule
PI-WGAN              & 0.136\,{\tiny$\pm$.010} & \textbf{1.000} & $+86.5\%$ & 0.0386 \\
PI-MMD-VAE           & 0.793\,{\tiny$\pm$.018} & \textbf{1.000} & $+20.7\%$ & 0.1724 \\
PI-DDPM              & 0.021\,{\tiny$\pm$.010} & \textbf{1.000} & $+97.9\%$ & 0.0110 \\
PI-FlowTS            & 0.022\,{\tiny$\pm$.008} & \textbf{1.000} & $+97.8\%$ & 0.0016 \\
PI-LSTM-WGAN         & 0.016\,{\tiny$\pm$.007} & \textbf{1.000} & $+98.4\%$ & 0.0032 \\
PI-LSTM-MMD-VAE      & 0.257\,{\tiny$\pm$.024} & \textbf{1.000} & $+74.3\%$ & 0.0665 \\
PI-MMD-VAE-WGAN      & 0.018\,{\tiny$\pm$.006} & \textbf{1.000} & $+98.2\%$ & 0.0067 \\
PI-LSTM-MMD-VAE-WGAN & 0.017\,{\tiny$\pm$.003} & \textbf{1.000} & $+98.3\%$ & 0.0093 \\
PI-TC-MMD-VAE        & 1.000 & \textbf{1.000} & $+0.0\%$  & 0.2437 \\
PI-RealNVP-TC-MMD-VAE& 1.000 & \textbf{1.000} & $+0.0\%$  & 0.2440 \\
\bottomrule
\end{tabular}
\endgroup
\end{table}

\textbf{TC covariance collapse and $\beta$-warm-up:} Static $\beta{=}6$ in RealNVP-TC-MMD-VAE produces below-chance attribution ($\kappa{=}-0.076$ on seed~42, mean $0.053\!\pm\!0.105$): the TC penalty severs the negative latent correlations that $\mathbf{c}{=}\mathbf{H}\boldsymbol{\delta}$ requires, leaving a generator that matches marginal distributions but not the jointly feasible attack manifold. The $50$-epoch $\beta$-warm-up (Section~\ref{ssec:beta_warmup}) separates two phases: the encoder first learns the joint structure (TC loss $\approx\!78$ at epoch $50$), then disentanglement is gradually activated (TC loss $\to\!0.007$ at epoch $150$). The schedule drives $\kappa{\to}0.785$ on 30-bus and reduces MMD by $14.1\%$ (TC-MMD-VAE) and $19.2\%$ (RealNVP-TC-MMD-VAE). On 14-bus AC, the same schedule moves TC-MMD-VAE from $\kappa{=}-0.036$ to $0.655\!\pm\!0.106$ ($\Delta\kappa{=}+0.691$). On 14-bus DC, neither architecture benefits because the rank-$13$ $\mathbf{H}$ does not produce the joint correlation pattern needed for a detectable warm-up shift.

\subsection{Block A: 14-bus AC (68-dimensional phasor)}
\label{ssec:results_ddetmtd}

On the 14-bus AC, every model achieves $\epsilon_{\mathrm{BDD}}{=}1.000$; the 35,130-sample, 68-dimensional corpus provides a sufficiently rich signal for every architecture to learn measurement patterns that are invisible to BDD. Table~\ref{tab:ddetmtd} reports the leading models sorted by sliced Wasserstein distance $W_1$. The three lowest-$W_1$ models are RealNVP-TC-MMD-VAE ($0.0104$), TC-MMD-VAE ($0.0107$), and MMD-VAE ($0.0108$), each with $\epsilon_{\mathrm{AE}}{=}0.935$. Physics-informed wrapping produces a second negative result on this testbed: PI- variants exhibit uniformly lower LSTM-AE evasion than their base counterparts (PI-RealNVP-TC-MMD-VAE: $0.777$ vs. $0.935$; PI-DDPM: $0.521$ vs. $0.908$; PI-FlowTS: $0.236$ vs. $0.907$.) The DC-derived projection introduces systematic deviations in the imaginary components of current phasors that the LSTM-AE learns to flag, recurring with greater severity than on 30-bus DC. The three testbeds, therefore, expose three distinct regimes of the physics--evasion interaction, none of which would be visible in a single-testbed evaluation. 
\begin{table}[!t]
\centering
\caption{Block A on 14-bus AC (68-dim), sorted by $W_1$. $\epsilon_{\mathrm{BDD}}{=}1.000$ for all models. Pareto-optimal in bold.}
\label{tab:ddetmtd}
\begingroup
\fontsize{6.5}{7.8}\selectfont
\setlength{\tabcolsep}{2pt}
\begin{tabular}{l*{5}{>{\fontsize{6.2}{7.5}\selectfont\arraybackslash}c}}
\toprule
\textbf{Model} & \textbf{$W_{1}$} & \textbf{MMD} & $\boldsymbol{\epsilon}_{\mathrm{AE}}$ & $\boldsymbol{\phi}$ & $\boldsymbol{\psi}$ \\
\midrule
\textbf{RealNVP-TC-MMD-VAE} & 0.0104\,{\tiny$\pm$.0003} & 0.101\,{\tiny$\pm$.010} & 0.935 & 0.094\,{\tiny$\pm$.006} & 0.911 \\
\textbf{TC-MMD-VAE}    & 0.0107\,{\tiny$\pm$.0002} & 0.123\,{\tiny$\pm$.007} & 0.935 & 0.081\,{\tiny$\pm$.013} & 0.911 \\
\textbf{MMD-VAE}       & 0.0108\,{\tiny$\pm$.0001} & 0.112\,{\tiny$\pm$.004} & 0.935 & 0.162\,{\tiny$\pm$.005} & 0.912 \\
LSTM-MMD-VAE          & 0.0117 & 0.074 & 0.934 & 0.161 & 0.905 \\
DDPM                  & 0.0155 & 0.010 & 0.908 & 0.146 & 0.902 \\
PI-TC-MMD-VAE         & 0.0169 & 0.055 & 0.873 & 0.318 & 0.910 \\
LSTM-MMD-VAE-WGAN     & 0.0188 & 0.017 & 0.930 & 0.166 & 0.906 \\
PI-RealNVP-TC-MMD-VAE & 0.0202 & 0.044 & 0.777 & 0.317 & 0.910 \\
FlowTS                & 0.0229 & 0.015 & 0.907 & 0.161 & 0.903 \\
LSTM-WGAN             & 0.0230 & 0.022 & 0.898 & 0.143 & 0.904 \\
PI-MMD-VAE            & 0.0233 & 0.035 & 0.784 & 0.299 & 0.908 \\
PI-LSTM-MMD-VAE       & 0.0299 & 0.082 & 0.621 & 0.273 & 0.906 \\
WGAN                  & 0.0371 & 0.081 & 0.739 & 0.129 & 0.906 \\
PI-DDPM               & 0.0448 & 0.169 & 0.521 & 0.272 & 0.899 \\
PI-LSTM-WGAN          & 0.0467 & 0.141 & 0.448 & 0.272 & 0.902 \\
PI-LSTM-MMD-VAE-WGAN  & 0.0496 & 0.158 & 0.470 & 0.269 & 0.907 \\
MMD-VAE-WGAN          & 0.0512 & 0.335 & 0.516 & 0.124 & 0.908 \\
PI-MMD-VAE-WGAN       & 0.0624 & 0.171 & 0.246 & 0.273 & 0.906 \\
PI-FlowTS             & 0.0632 & 0.171 & 0.236 & 0.277 & 0.904 \\
PI-WGAN               & 0.0738 & 0.221 & 0.252 & 0.261 & 0.905 \\
\bottomrule
\end{tabular}
\endgroup
\end{table}

\subsection{Statistical validation of pool diversity and physics effects}
\label{ssec:results_stats}

Friedman applied to the $20\!\times\!25$ MMD matrix yields $\chi^{2}_{F}(19){=}459.13$, $p{=}2.03{\times}10^{-85}$, and to BDD evasion $\chi^{2}_{F}(19){=}380.00$, $p{=}6.26{\times}10^{-69}$. Nemenyi post-hoc with Bonferroni correction identifies $20$ MMD-distinct and $19$ evasion-distinct pairs. The ten paired Wilcoxon signed-rank tests for PI/base differences are all significant at $p\!\le\!1.13{\times}10^{-6}$; the sign of Cohen's $d$ depends on family. PI-wrapping reduces MMD for five VAE-based pairs ($d\!\in\![-7.48,-0.66]$) and increases it for five DDPM, flow, and LSTM-GAN pairs ($d\!\in\![+0.95,+7.74]$).

\begin{figure}[!t]
\centering
\includegraphics[width=\columnwidth]{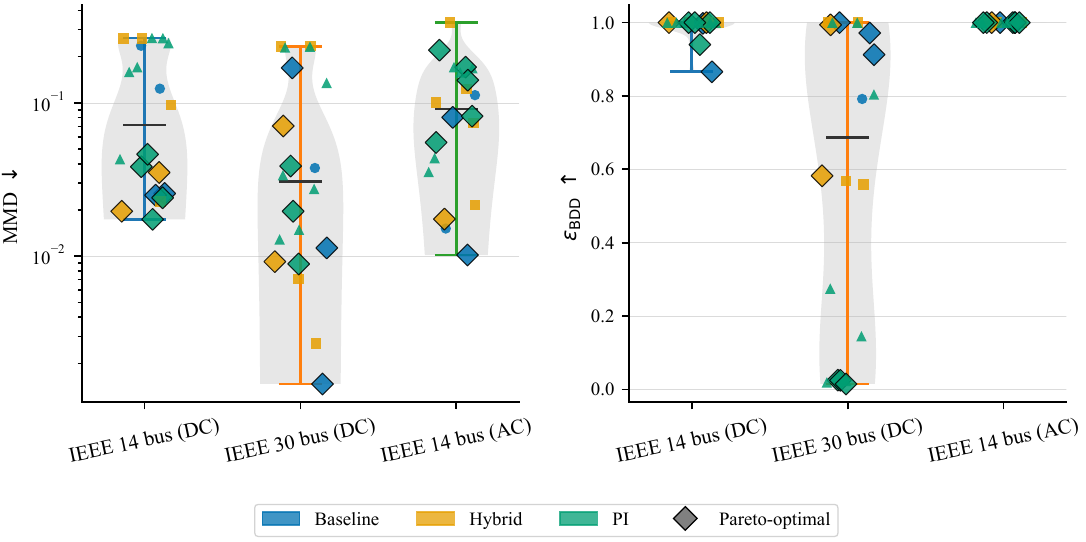}
\caption{Pool-level MMD (left) and BDD-evasion (right) distributions across testbeds, grouped by family (Baseline, Hybrid, PI). Diamonds mark Pareto-optimal models, the median is shown by horizontal bars. The 30-bus DC violin reveals the bimodal $\epsilon_{\mathrm{BDD}}$ structure produced by the failure mode of Lemma~\ref{lem:mean_leak}, eight PI- variants collapse near zero while two TC-based PI variants and the bases remain near unity.}
\label{fig:figR5}
\end{figure}

Fig.~\ref{fig:figR5} renders these family-level differences as the underlying distributions across testbeds. The MMD violin (left, log scale) is widest on 30-bus DC, where the lower tail extends below $10^{-3}$ for FlowTS and the LSTM-bearing hybrids, while the upper tail still reaches ${\approx}3{\times}10^{-1}$ for the high-MMD MMD-VAE-WGAN, spanning more than two decades. The 14-bus DC spread is narrower, concentrated between $2{\times}10^{-2}$ and $3{\times}10^{-1}$, and the 14-bus AC distribution is shifted upward and tightly clustered around $5{\times}10^{-2}$ to $2{\times}10^{-1}$, reflecting that the 68-dimensional phasor corpus produces uniformly higher absolute MMD than the lower-dimensional DC testbeds because the metric scales with feature dimension. The BDD-evasion violin (right) makes the failure mode of Lemma~\ref{lem:mean_leak} visible at a glance. The 30-bus DC distribution is bimodal, with eight PI variants concentrated near zero and the two TC-based PI variants together with the bases pinned at unity, whereas the 14-bus DC and 14-bus AC remain unimodal at $\epsilon_{\mathrm{BDD}}{=}1$. Pareto-optimal models (diamonds) lie at the lower tail of the MMD violin in every testbed, confirming that the four-objective selection of Section~\ref{ssec:conditioning} consistently identifies the distributionally tightest fits. Fig.~\ref{fig:figR6} shows the complete $20\!\times\!9$ evasion heatmap consolidating all three testbeds with the relevant per-detector evasion rates.

\begin{figure}[!t]
\centering
\includegraphics[width=\columnwidth]{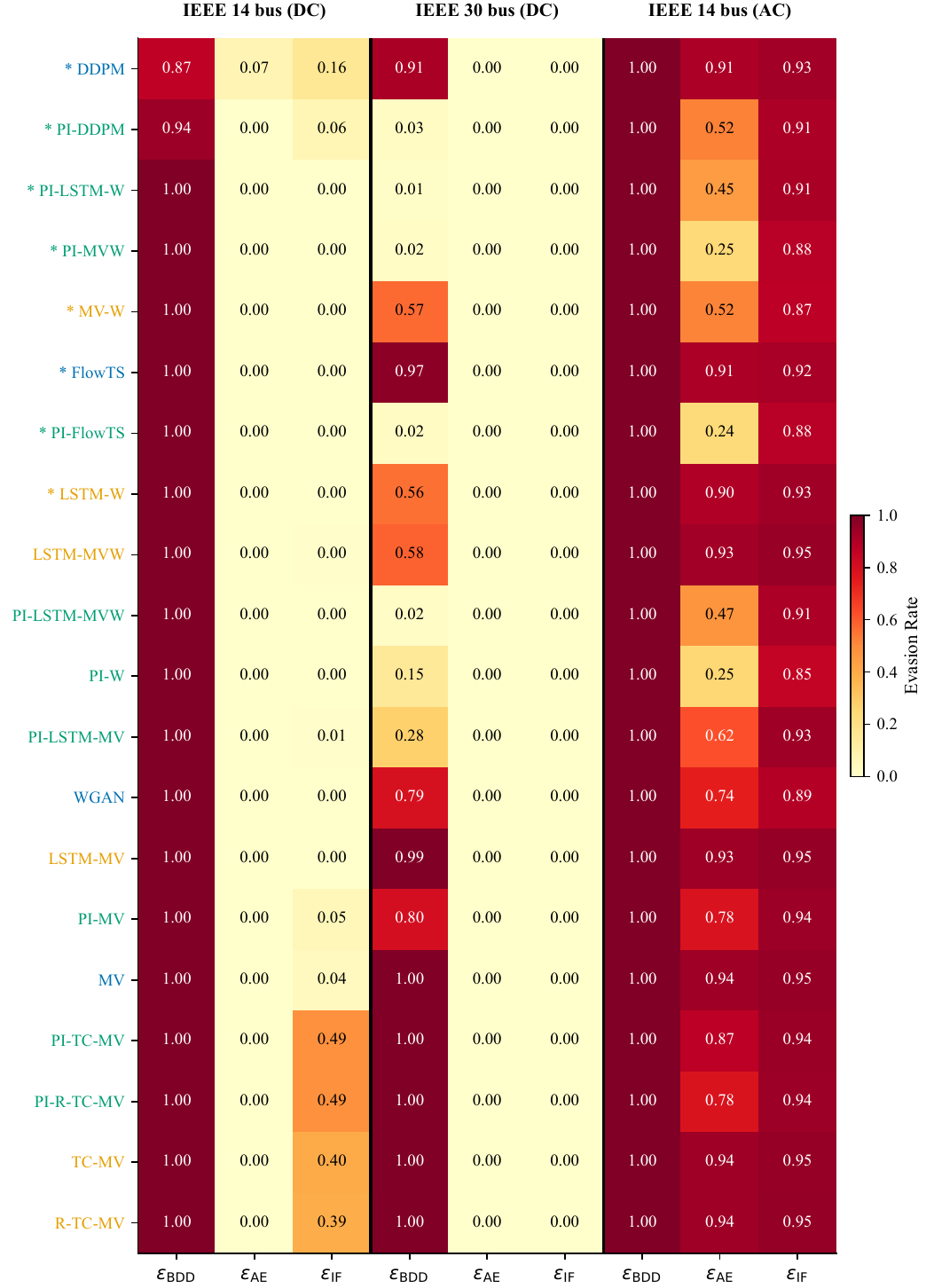}
\caption{Evasion heatmap: 20 pool models (sorted by Pareto rank, starred entries are Pareto-optimal) versus nine detector--testbed combinations. Cell values are mean evasion rates over 25 runs (Acronyms (for hybrid models, e.g., MV: MMD-VAE): M: MMD, W: WGAN, V: VAE.)}
\label{fig:figR6}
\end{figure}

\subsection{Block B: attacker-knowledge conditioning}
\label{ssec:results_blockb}

Table~\ref{tab:blockb_wd} reports $W_1$ across the five $k$-levels for the representative per-family triplet on each testbed; Fig.~\ref{fig:figR3} renders the corresponding ribbons with Kruskal--Wallis statistics and MITRE ATT\&CK stage labels. Across all nine model--testbed pairs, the Kruskal--Wallis $H$ exceeds $16.7$ at $p\!\le\!0.0022$ with $\eta^{2}\!\ge\!0.636$, demonstrating strong knowledge-driven fidelity differences. The increase is monotonic in $k$ in every case up to $k{=}0.75$, with a small reversal at $k{=}1$ on the 14-bus AC for two of the three models, attributable to the over-constraint of the affine projection when the entire $\mathbf{H}$ is supplied as conditioning input. Physics consistency $\phi$ peaks at $k{=}0.75$ and drops at $k{=}1$ on every testbed (bottom row of Fig.~\ref{fig:figR3}). The greatest degradation in relative fidelity is on the 30-bus DC: FlowTS increases $37.8\times$ in $W_1$ from $k{=}0$ to $k{=}1$, and MMD-VAE increases $49.2\times$. The $k{=}0$ regime represents the most realistic external adversary, while $k{=}1$ should be reserved for worst-case insider scenarios.

\begin{table}[!t]
\centering
\caption{Block B: $W_1$ (mean$\pm$SD, 5 seeds) by $k$ for the representative per-family triplet on each testbed (MV: MMD-VAE, W: WGAN).}
\label{tab:blockb_wd}
\begingroup
\fontsize{6.2}{7.4}\selectfont
\setlength{\tabcolsep}{1.8pt}
\begin{tabular}{ll*{5}{>{\fontsize{6}{7.2}\selectfont\arraybackslash}r}}
\toprule
\textbf{Testbed} & \textbf{Model} & $k{=}0$ & $k{=}0.25$ & $k{=}0.5$ & $k{=}0.75$ & $k{=}1$ \\
\midrule
\multirow{3}{*}{14-bus DC}
 & DDPM        &  10.20\,{\tiny$\pm$1.35} &  12.40\,{\tiny$\pm$1.48} &  14.46\,{\tiny$\pm$2.35} &  15.11\,{\tiny$\pm$1.18} &  18.60\,{\tiny$\pm$1.83} \\
 & PI-DDPM     &   2.34\,{\tiny$\pm$0.79} &   3.12\,{\tiny$\pm$0.62} &   5.61\,{\tiny$\pm$2.48} &   6.52\,{\tiny$\pm$2.09} &   8.15\,{\tiny$\pm$1.66} \\
 & PI-LSTM-W   &   0.27\,{\tiny$\pm$0.02} &   0.79\,{\tiny$\pm$0.08} &   1.40\,{\tiny$\pm$0.13} &   2.11\,{\tiny$\pm$0.15} &   2.69\,{\tiny$\pm$0.13} \\
\midrule
\multirow{3}{*}{30-bus DC}
 & MMD-VAE     &   1.83\,{\tiny$\pm$0.17} &  22.67\,{\tiny$\pm$1.20} &  45.01\,{\tiny$\pm$1.65} &  66.78\,{\tiny$\pm$3.32} &  89.74\,{\tiny$\pm$3.54} \\
 & FlowTS      &   9.30\,{\tiny$\pm$0.62} &  94.91\,{\tiny$\pm$3.52} & 177.98\,{\tiny$\pm$3.44} & 265.86\,{\tiny$\pm$9.60} & 351.54\,{\tiny$\pm$8.64} \\
 & PI-DDPM     &  42.19\,{\tiny$\pm$6.21} & 110.60\,{\tiny$\pm$13.0} & 204.30\,{\tiny$\pm$29.6} & 281.23\,{\tiny$\pm$18.7} & 431.49\,{\tiny$\pm$12.2} \\
\midrule
\multirow{3}{*}{14-bus AC}
 & WGAN        & 0.0080\,{\tiny$\pm$.001} & 0.0090\,{\tiny$\pm$.001} & 0.0123\,{\tiny$\pm$.001} & 0.0169\,{\tiny$\pm$.001} & 0.0155\,{\tiny$\pm$.001} \\
 & DDPM        & 0.0047\,{\tiny$\pm$.001} & 0.0049\,{\tiny$\pm$.000} & 0.0060\,{\tiny$\pm$.000} & 0.0078\,{\tiny$\pm$.000} & 0.0073\,{\tiny$\pm$.001} \\
 & PI-WGAN     & 0.0165\,{\tiny$\pm$.001} & 0.0172\,{\tiny$\pm$.003} & 0.0182\,{\tiny$\pm$.002} & 0.0240\,{\tiny$\pm$.001} & 0.0589\,{\tiny$\pm$.002} \\
\bottomrule
\end{tabular}
\endgroup
\end{table}

\begin{figure*}[!t]
\centering
\includegraphics[width=.95\textwidth]{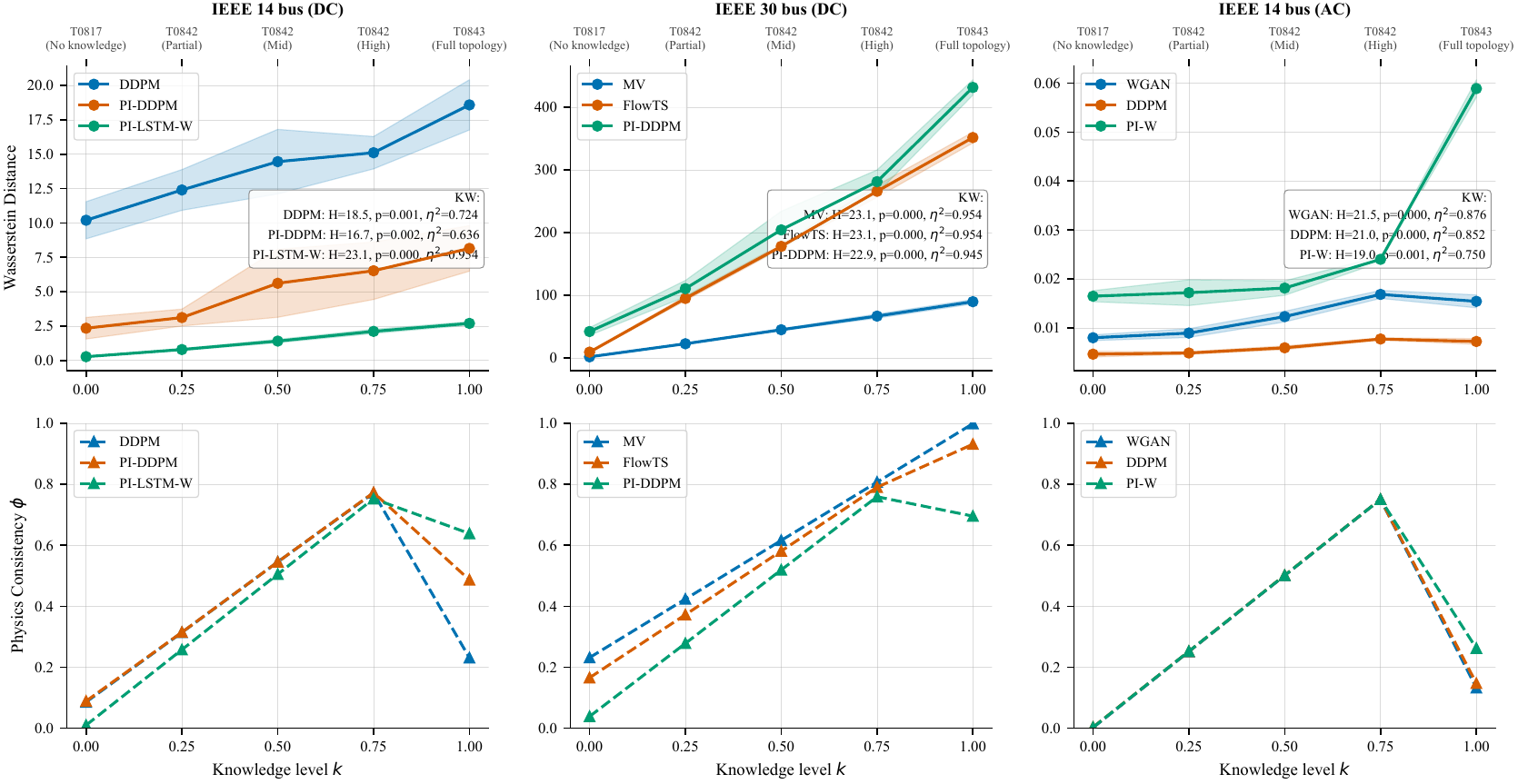}
\caption{Block B conditioning ribbons: $W_1$ (top) and physics consistency $\phi$ (bottom) versus $k$ for the representative per-family triplet on each testbed. MITRE ATT\&CK stages annotate the top axis. $\phi$ peaks at $k{=}0.75$ on every testbed (Acronyms (for hybrid models, e.g., MV: MMD-VAE): M: MMD, W: WGAN, V: VAE.).}
\label{fig:figR3}
\end{figure*}

\subsection{Ablation summary}
\label{ssec:results_ablation}

The four ablation conditions (\emph{no} physics, \emph{no} selection, \emph{no} conditioning, \emph{no} temporal) were trained from scratch with three seeds for $200$ epochs on three PI-probes across all testbeds. The two non-trivial conditions are summarised here; \emph{no-selection} and \emph{no} conditioning produce $\Delta\mathrm{MMD}\!\approx\!0$ at the $k{=}0$ default. \emph{No-physics} on the DC testbeds increases MMD by $26\%$ (14-bus) and $18\%$ (30-bus) for PI-LSTM-MMD-VAE-WGAN, while leaving the TC and flow probes nearly unchanged. On 14-bus AC, the sign reverses: removing the physics projection \emph{reduces} MMD by $78\%$ for PI-LSTM-MMD-VAE-WGAN and by $32\%$ for PI-RealNVP. The DC-derived $\mathbf{H}$ embedded in \textsc{PhysicsWrapper} introduces a systematic mismatch under 68-dimensional AC phasor measurements, and disabling it improves fidelity. The \emph{no-temporal} effect produces the largest single-component impact: removing the LSTM encoder increases MMD by $294\%$ on 14-bus DC and $591\%$ on 30-bus DC for PI-LSTM-MMD-VAE-WGAN. The same removal is negligible on 14-bus AC, where the 35{,}130-sample corpus provides sufficient signal for non-temporal architectures (Fig.~\ref{fig:figR4}).

\begin{figure*}[!t]
\centering
\includegraphics[width=\textwidth]{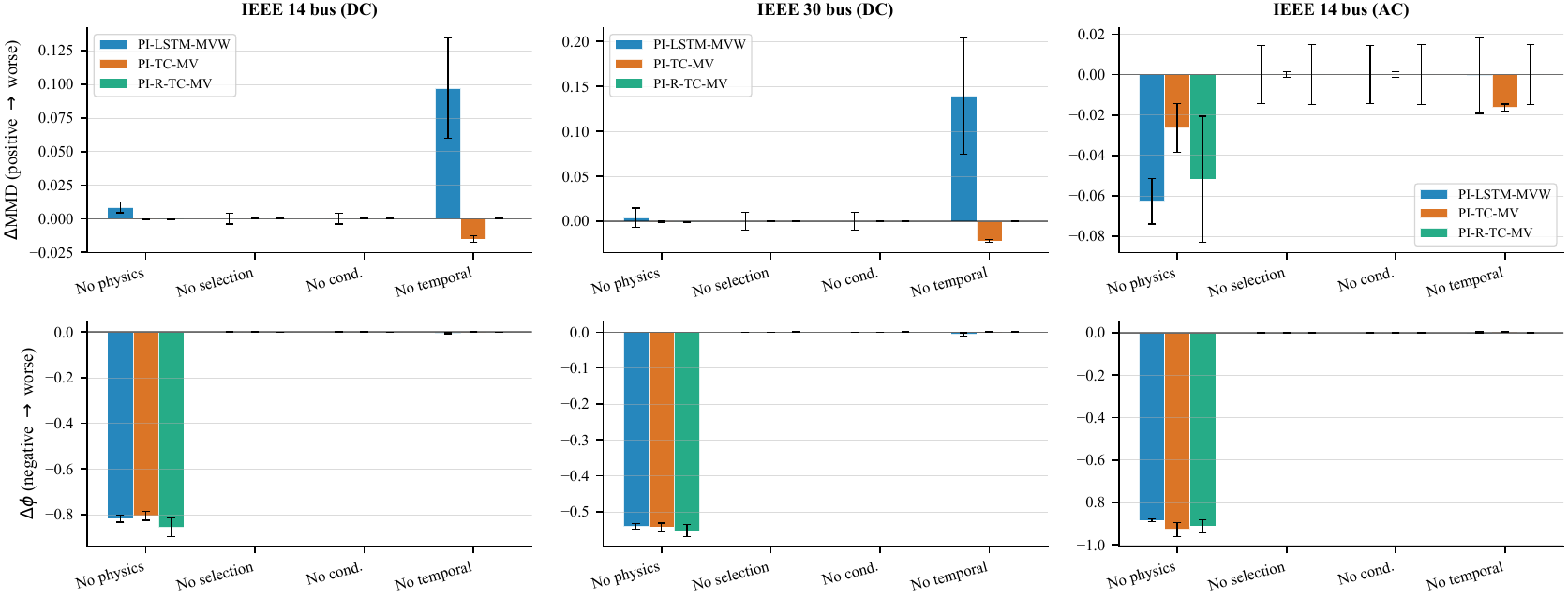}
\caption{Ablation waterfall: $\Delta\mathrm{MMD}$ relative to the full model for four ablation conditions across three PI- probes and three testbeds. Note the sign reversal of \emph{no-physics} on 14-bus AC (Acronyms (for hybrid models, e.g., MV: MMD-VAE): M: MMD, W: WGAN, V: VAE.).}
\label{fig:figR4}
\end{figure*}

\subsection{Block C: cross-level explainability}
\label{ssec:results_xai}

Cohen's $\kappa$ on the top-5 sensor rankings is reported in Table~\ref{tab:xai_kappa} after the harmoniser is applied to all PI- models and TC checkpoints are loaded from $\beta$-warm-up. PI-LSTM-MMD-VAE reaches $\kappa{=}1.000$ on the 14-bus DC, and both PI-TC-MMD-VAE and PI-LSTM-MMD-VAE-WGAN reach $1.000$ on the 30-bus DC. On both DC testbeds, the PI- mean exceeds the base mean: $0.706$ vs.\ $0.559$ on the 14-bus, $0.785$ vs.\ $0.139$ on the 30-bus. AC68 is reported as additional evidence: PI- $0.568$ vs.\ base $0.137$. Fig.~\ref{fig:recovery} visualises the joint impact of the harmoniser on BDD evasion (left panel, all ten PI- models) and on cross-level agreement under the four configurations of PI-TC-MMD-VAE (right panel).

\begin{table}[!t]
\centering
\caption{Block C: Cohen's $\kappa$ (data- vs.\ model-level top-5 sensor agreement). PI- models use the harmoniser; TC models use $\beta$-warm-up checkpoints.}
\label{tab:xai_kappa}
\begingroup
\fontsize{7}{8.4}\selectfont
\setlength{\tabcolsep}{3pt}
\begin{tabular}{lccc}
\toprule
\textbf{Model} & \textbf{14-bus} & \textbf{30-bus} & \textbf{AC68} \\
\midrule
PI-LSTM-MMD-VAE-WGAN  & \textbf{1.000} & 1.000 & 0.784 \\
PI-TC-MMD-VAE         & 0.559          & \textbf{1.000} & 0.352 \\
PI-RealNVP-TC-MMD-VAE & 0.559          & 0.355 & 0.568 \\
LSTM-MMD-VAE-WGAN     & 0.559          & 0.139 & 0.137 \\
\bottomrule
\end{tabular}
\endgroup
\end{table}

\begin{figure}[!t]
\centering
\includegraphics[width=\columnwidth]{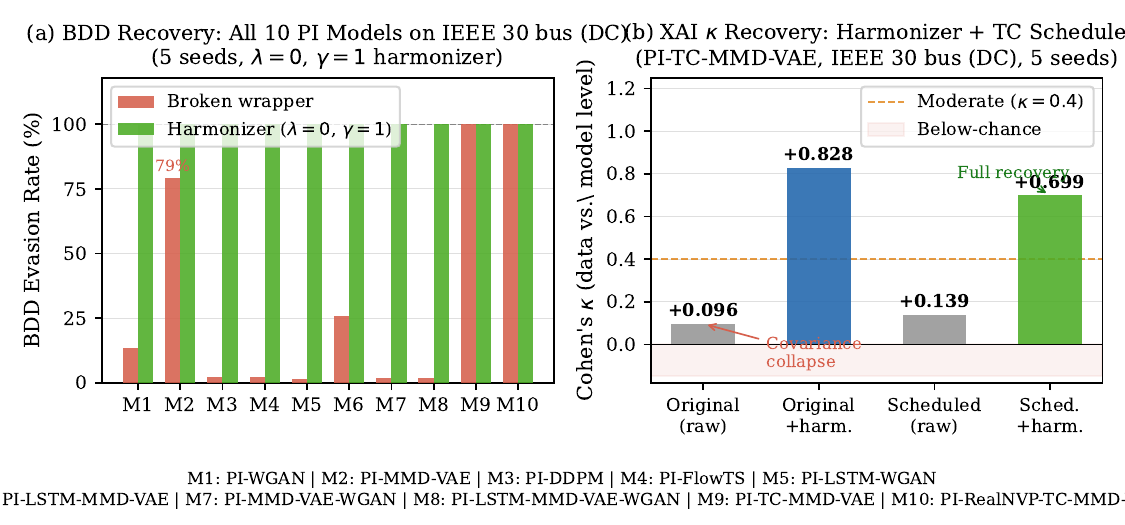}
\caption{Block C joint recovery on 30-bus DC. \emph{Left}: BDD evasion before (broken \textsc{PhysicsWrapper}, normalised-space projection) and after (harmoniser) for all ten PI- models; eight of ten collapse to ${<}30\%$ before harmonisation and recover to $100\%$ after. \emph{Right}: Cohen's $\kappa$ for PI-TC-MMD-VAE under four configurations, covariance collapse with raw original ($\kappa{=}+0.096$), harmoniser alone ($+0.828$), $\beta$-warm-up alone ($+0.139$), and combined ($+0.699$, well above the moderate-agreement threshold $\kappa{=}0.4$).}
\label{fig:recovery}
\end{figure}

The mechanism is the $\operatorname{col}(\mathbf{H})$ projection. By forcing all attack energy into the jointly observable subspace, the harmoniser makes the projection magnitude a coherent model-level signal that aligns with the data-level Wasserstein shift on the dominant sensors. Attacks with residual null-space components produce incoherent model-level gradients because the null-space contribution is undetectable at the data level but appears in the IG path integral. A null space-aware blended projection $\hat{\mathbf{a}}_{i}=t_{i}\mathbf{P}_{\!\mathbf{H}}\mathbf{a}_{i}+(1-t_{i})\mathbf{a}_{i}$ provides a clean diagnostic: $\kappa$ jumps discontinuously from $0.139$ to $0.398$ only at $t{=}1$, while MMD remains flat across all $t$, showing that the column space and null space carry disjoint attribution signatures, and partial enforcement is dominated by the full harmoniser. Prediction-level attribution applied to the final-layer logit produced zero-valued gradients, a known artefact of binary classifiers operating at high confidence, and is therefore omitted.

\subsection{Comparison with prior work}
\label{ssec:disc_related}

\textit{iAttackGen} \cite{shahriar2022iattackgen} reports $>\!90\%$ BDD evasion for a single WGAN on 14-bus DC without a holdout test partition. Under the strict protocol used here, WGAN reaches comparable evasion ($\epsilon_{\mathrm{BDD}}{=}1.000$) but is Pareto-dominated by PI-DDPM (lower MMD) and LSTM-MMD-VAE-WGAN (higher Isolation Forest evasion), justifying the multi-model pool. \textit{SparseGAN} \cite{li2023sparsegan} and related constraint-based approaches \cite{li2024advdiffusionfdia} omit the explicit physics projection; the no-physics ablation increases MMD by up to $26\%$ on 14-bus and $18\%$ on 30-bus for LSTM-bearing PI- models, supporting projection as a regulariser even on testbeds where it carries an evasion penalty. Prior work reporting $>\!90\%$ BDD evasion for PI- models on 30-bus networks operates under two compounding artefacts: clean measurements in per-unit (mean ${\sim}0.19$) evaluated against attack data in MW (mean ${\sim}322$), creating a $1615\times$ scale mismatch that renders the BDD threshold meaningless; and distributional metrics computed against the training partition rather than a held-out evaluation set. The strict 60/20/20 chronological partition with data-driven threshold calibration gives the first evaluation in which the geometric conflict is visible, and the harmoniser closes the resulting performance gap to $100\%$ without retraining. Concurrent diffusion-based FDIA work \cite{li2024advdiffusionfdia,zhang2024piddpm,qin2026false} reports BDD evasion only; the multi-detector probe used here exposes a vulnerability surface that single-detector evaluations conceal: on the 14-bus AC, the best base model achieves $\epsilon_{\mathrm{AE}}{=}0.935$, while PI-variants drop to $0.236$--$0.777$.

\subsection{Limitations}
\label{ssec:disc_limitations}

\textsc{PhysicsWrapper} uses the DC power-flow linearisation, which ignores reactive power, voltage-magnitude variation, and transformer tap positions; the 14-bus AC testbed already shows that the DC projection corrupts imaginary phasor components. All three testbeds use simulator output; field PMU recordings were unavailable due to commercial confidentiality. LSTM-AE evasion is zero for all 20 models on 30-bus DC, reflecting the more distinctive measurement-correlation structure of that topology rather than a generator deficiency. A shadow-LSTM-AE fine-tuning experiment reduced shadow reconstruction error on attacks by ${\sim}600\times$ in 20 epochs but moved BDD evasion only marginally, indicating that threshold-based BDD is insensitive to the distributional shifts achievable in a short fine-tuning window; a differentiable surrogate BDD loss with curriculum difficulty is needed for reliable end-to-end evasion optimisation.

%% ==========================================================================
\section{Conclusion}
\label{sec:conclusion}

\textsc{GenAI-FDIA} is a unified framework for benchmarking physics-informed generative FDIA synthesis. A pool of $P{=}20$ architectures spanning Wasserstein GANs, MMD-VAE, normalising flows, and score-based diffusion, including four novel cross-family hybrids (MMD-VAE-WGAN, LSTM-MMD-VAE-WGAN, TC-MMD-VAE, RealNVP-TC-MMD-VAE), is evaluated on three IEEE benchmarks (14-bus DC, 30-bus DC, 14-bus AC with 68-dimensional phasor measurements) under a strict 60/20/20 chronological partition with data-driven detector calibration. The four hybrids each yield a measurable gain over their parents; MMD-VAE-WGAN (MMD $0.035$ on 14-bus DC) is $6.7\times$ tighter than the MMD-VAE base ($0.236$) and $3.5\times$ tighter than WGAN ($0.124$), confirming that the adversarial discriminator sharpens the encoder posterior; LSTM-MMD-VAE-WGAN is Pareto-optimal on 14-bus DC at MMD $0.023$, a $4.2\times$ improvement over LSTM-MMD-VAE ($0.097$); RealNVP-TC-MMD-VAE leads the 14-bus AC frontier ($W_1{=}0.0104$, the lowest of all 20 models) ahead of TC-MMD-VAE ($0.0107$) and the MMD-VAE parent ($0.0108$), and the addition of total-correlation disentanglement to the latent yields $\kappa{=}1.000$ cross-level agreement on 30-bus DC after $\beta$-warm-up.

The main contributions are: (i) a multi-model Pareto pool whose non-dominated subset varies across testbeds (PI-DDPM, LSTM-WGAN, and LSTM-MMD-VAE-WGAN on 14-bus DC; FlowTS, PI-FlowTS, and PI-LSTM-WGAN on 30-bus DC after the harmoniser; RealNVP-TC-MMD-VAE, TC-MMD-VAE, and MMD-VAE on 14-bus AC), supporting the rejection of a single architectural choice; (ii) a formal geometric characterisation of the $\boldsymbol{\mu}\!\notin\!\operatorname{col}(\mathbf{H})$ failure mode that displaces inverse-scaled physics-projected attacks outside the stealthy subspace, together with an inference-time harmoniser that restores $\epsilon_{\mathrm{BDD}}{=}100\%$ for all ten PI- variants without retraining; (iii) the identification of a total-correlation covariance-collapse mode and its resolution via a $50$-epoch $\beta$-warm-up schedule that lifts cross-level attribution agreement from below-chance to substantial; and (iv) a continuous attacker-knowledge axis $k\!\in\![0,1]$ mapped to MITRE ATT\&CK ICS sub-techniques, under which $W_1$ increases monotonically by $1.6$--$49.2\times$ from $k{=}0$ to $k{=}1$ at strong effect sizes ($\eta^{2}\!\ge\!0.636$).

Empirically, every pool member achieves $\epsilon_{\mathrm{BDD}}\!\ge\!86.6\%$ on the 14-bus DC; the residual test can be defeated by every generative family. Physics-informed wrapping raises $\phi$ by $2.6\times$ on the 14-bus, but under the broken normalised-space projection collapses BDD evasion to $<\!30\%$ for eight of ten PI-variants on the 30-bus; the harmoniser restores $100\%$ for all ten on every testbed. Temporal modelling is the dominant ablation driver ($\Delta\mathrm{MMD}$ up to $+591\%$ on the 30-bus DC when LSTM layers are removed). Cross-level XAI agreement reaches $\kappa{=}1.000$ on both DC testbeds with the harmoniser and $\beta$-warm-up applied. For practitioners, we recommend PI-LSTM-WGAN ($\epsilon_{\mathrm{BDD}}{=}100\%$, best physics consistency) or PI-DDPM (lowest MMD) as primary candidates on the 14-bus DC. For larger networks, PI-variants must apply the harmoniser in physical units rather than in normalised feature space. Three directions follow directly: a differentiable AC power-flow constraint layer to remove the DC approximation, an end-to-end evasion-training scheme via a differentiable surrogate BDD loss, and validation on field PMU recordings.

\section*{Acknowledgment}
The authors acknowledge the use of AI technologies (Claude Opus 4.7) during the preparation of this manuscript strictly for English language checking, \LaTeX\ code error correction, and document formatting. The authors retain full responsibility for the scientific content, accuracy, and integrity of this work.

%% -- Bibliography -----------------------------------------------------------
\bibliographystyle{IEEEtran}
\bibliography{references}

@techreport{irena2024renewables,
  author      = {{International Renewable Energy Agency}},
  title       = {Renewable Power Generation Costs in 2023},
  institution = {IRENA},
  year        = {2024},
  url         = {https://www.irena.org/Publications/2024/Sep/Renewable-Power-Generation-Costs-in-2023}
}

@techreport{lee2016ukrainegrid,
  author      = {Lee, Robert M. and Assante, Michael J. and Conway, Tim},
  title       = {Analysis of the Cyber Attack on the {Ukrainian} Power Grid},
  institution = {E-ISAC and SANS ICS},
  year        = {2016},
  month       = mar,
  url         = {https://ics.sans.org/media/E-ISAC_SANS_Ukraine_DUC_5.pdf}
}

@techreport{eset2022industroyer2,
  author      = {{ESET Research}},
  title       = {{INDUSTROYER2}: {Industroyer} Reloaded},
  institution = {ESET},
  year        = {2022},
  month       = apr,
  url         = {https://www.welivesecurity.com/2022/04/12/industroyer2-industroyer-reloaded/}
}

@techreport{cisa2024volttyphoon,
  author      = {{CISA} and {NSA} and {FBI}},
  title       = {{People's Republic of China} State-Sponsored Cyber Actor
                 Living off the Land to Evade Detection},
  institution = {Cybersecurity and Infrastructure Security Agency},
  year        = {2024},
  month       = feb,
  url         = {https://www.cisa.gov/news-events/cybersecurity-advisories/aa24-038a}
}

@article{liu2011fdia,
  author    = {Liu, Yao and Ning, Peng and Reiter, Michael K.},
  title     = {False Data Injection Attacks against State Estimation
               in Electric Power Grids},
  journal   = {ACM Transactions on Information and System Security},
  volume    = {14},
  number    = {1},
  pages     = {1--33},
  year      = {2011},
  month     = may,
  doi       = {10.1145/1952982.1952995}
}

@article{tian2022afdias,
  author    = {Tian, Jiwei and Wang, Buhong and Wang, Zhen and
               Cao, Kunrui and Li, Jing and Ozay, Mete},
  title     = {Joint Adversarial Example and False Data Injection
               Attacks for State Estimation in Power Systems},
  journal   = {IEEE Transactions on Cybernetics},
  volume    = {52},
  number    = {12},
  pages     = {13699--13713},
  year      = {2022},
  month     = dec,
  doi       = {10.1109/TCYB.2021.3125345}
}

@inproceedings{shahriar2022iattackgen,
  author    = {Shahriar, Md Hasan and Khalil, Alvi Ataur and
               Rahman, Mohammad Ashiqur and Manshaei, Mohammad Hossein
               and Chen, Dong},
  title     = {{iAttackGen}: Generative Synthesis of False Data
               Injection Attacks in Cyber-Physical Systems},
  booktitle = {Proceedings of the IEEE International Conference on
               Communications, Control, and Computing Technologies
               for Smart Grids (SmartGridComm)},
  year      = {2022},
  doi       = {10.1109/SmartGridComm52983.2022.9961004}
}

@article{xu2023ddetmtd,
  author    = {Xu, Wangkun and Higgins, Martin and Wang, Jianhong and
               Jaimoukha, Imad M. and Teng, Fei},
  title     = {Blending Data and Physics Against False Data Injection
               Attack: An Event-Triggered Moving Target Defence Approach},
  journal   = {IEEE Transactions on Smart Grid},
  volume    = {14},
  number    = {4},
  pages     = {2613--2627},
  year      = {2023},
  month     = jul,
  doi       = {10.1109/TSG.2022.3231728}
}

@article{musleh2023lstmfdia,
  author    = {Musleh, Ahmed S. and Chen, Guo and Dong, Zhao Yang and
               Wang, Chen and Chen, Shiping},
  title     = {Spatio-Temporal Data-Driven Detection of False Data
               Injection Attacks in Power Distribution Systems},
  journal   = {International Journal of Electrical Power \& Energy Systems},
  volume    = {147},
  pages     = {108926},
  year      = {2023},
  doi       = {10.1016/j.ijepes.2022.108926}
}

@inproceedings{gulrajani2017wgangp,
  author    = {Gulrajani, Ishaan and Ahmed, Faruk and Arjovsky, Martin
               and Dumoulin, Vincent and Courville, Aaron},
  title     = {Improved Training of {Wasserstein} {GANs}},
  booktitle = {Advances in Neural Information Processing Systems},
  volume    = {30},
  year      = {2017}
}

@inproceedings{tolstikhin2018wae,
  author    = {Tolstikhin, Ilya and Bousquet, Olivier and Gelly, Sylvain
               and Sch\"{o}lkopf, Bernhard},
  title     = {{Wasserstein} Auto-Encoders},
  booktitle = {International Conference on Learning Representations},
  year      = {2018}
}

@inproceedings{chen2018tcvae,
  author    = {Chen, Ricky T. Q. and Li, Xuechen and Grosse, Roger and
               Duvenaud, David},
  title     = {Isolating Sources of Disentanglement in Variational
               Autoencoders},
  booktitle = {Advances in Neural Information Processing Systems},
  volume    = {31},
  year      = {2018}
}

@inproceedings{dinh2017realnvp,
  author    = {Dinh, Laurent and Sohl-Dickstein, Jascha and Bengio, Samy},
  title     = {Density Estimation Using Real-Valued Non-Volume Preserving
               ({Real NVP}) Transformations},
  booktitle = {International Conference on Learning Representations},
  year      = {2017}
}

@inproceedings{ho2020ddpm,
  author    = {Ho, Jonathan and Jain, Ajay and Abbeel, Pieter},
  title     = {Denoising Diffusion Probabilistic Models},
  booktitle = {Advances in Neural Information Processing Systems},
  volume    = {33},
  pages     = {6840--6851},
  year      = {2020}
}

@misc{hu2025flowts,
  author    = {Hu, Yang and Wang, Xiao and Ding, Zezhen and Wu, Lirong
               and Zhang, Huatian and Li, Stan Z. and Wang, Sheng and
               Zhang, Jiheng and Li, Ziyun and Chen, Tianlong},
  title     = {{FLOWTS}: Time Series Generation via Rectified Flow},
  year      = {2025},
  note      = {Preprint. Code: \url{https://github.com/UNITES-Lab/FlowTS}}
}

@article{tian2026admm,
  author  = {Tian, Jiwei and Shen, Chao and Lin, Chenhao and Zhang, Meng and
             Xia, Xiaofang and Ren, Chao and Zhang, Yuqing},
  title   = {{ADMM}-Based Adversarial False Data Injection Attacks Against
             Multi-Label Locational Detection},
  journal = {IEEE Transactions on Dependable and Secure Computing},
  year    = {2026},
  doi     = {10.1109/TDSC.2025.3605689}
}

@article{liu2025spatiotemporal,
  author  = {Liu, Chensheng and Li, Yongyu and Yu, Yang and Wu, Xiaoming and
             Yang, Ming and Tang, Yang and Long, Bo},
  title   = {Spatiotemporal Stealthy Attacks in Power Systems With
             High-Penetrated Renewable Energy Sources},
  journal = {IEEE Internet of Things Journal},
  year    = {2025},
  doi     = {10.1109/JIOT.2025.3612101}
}

@article{liu2025adtrans,
  author  = {Liu, Yulin and Shi, Libao},
  title   = {A Joint {FDIA} Detection and Localization Scheme and Experimental
             Analysis Based on Adversarial Anomaly Transformer for Power
             Distribution System},
  journal = {IEEE Transactions on Industrial Informatics},
  year    = {2025},
  doi     = {10.1109/TII.2025.3609087}
}

@article{qin2026false,
  title={False Data Injection Attacks Against UAV Battery State of Health Estimation with A Capacity-informed Latent Diffusion Model},
  author={Qin, Yan and Luo, Zhiqing and Che, Yunhong and Cao, Liang},
  journal={Green Energy and Intelligent Transportation},
  pages={100404},
  year={2026},
  publisher={Elsevier}
}

@article{li2023sparsegan,
  author    = {Li, Fengyong and Shen, Weicheng and Bi, Zhongqin and
               Su, Xiangjing},
  title     = {Sparse Adversarial Learning for {FDIA} Attack Sample
               Generation in Distributed Smart Grids},
  journal   = {Computer Modeling in Engineering \& Sciences},
  year      = {2023},
  doi       = {10.32604/cmes.2023.044431}
}

@inproceedings{liu2024lsganfdia,
  author    = {Liu, Yichen Henry and Liu, Bo and Liu, Xuebo and Wu, Hongyu},
  title     = {False Data Injection Attacks Based on Least Squares
               Generative Adversarial Networks with Reconstruction Loss},
  booktitle = {Proc.\ IEEE Power \& Energy Society General Meeting (PESGM)},
  year      = {2024},
  pages     = {1--5},
  month     = jul,
  doi       = {10.1109/pesgm51994.2024.10688650}
}

@inproceedings{hong2023tsganfdia,
  author    = {Hong, Sheng and Zhang, Xiao and Kou, Jian and Sun, Yuhao},
  title     = {A Novel False Data Injection Method Targeting on
               Time-series analysis in Smart Grid},
  booktitle = {Proc.\ Int.\ Conf.\ Smart Energy Grid Engineering (SEGE)},
  pages     = {50--55},
  year      = {2023},
  month     = jun,
  doi       = {10.1109/segre58867.2023.00016}
}

@article{huang2022sparseae,
  author    = {Huang, Xiaoge and Qin, Zhijun and Xie, Ming and Liu, Hui
               and Meng, Liang},
  title     = {Defense of Massive False Data Injection Attack via Sparse
               Attack Points Considering Uncertain Topological Changes},
  journal   = {Journal of Modern Power Systems and Clean Energy},
  volume    = {10},
  number    = {6},
  pages     = {1588--1598},
  year      = {2022},
  month     = jan,
  doi       = {10.35833/mpce.2020.000686}
}

@article{nawaz2025gangridfdia,
  author    = {Nawaz, Rehan and Akhtar, Rabbaya and Khan, Saad Ullah and
               Bu, Siqi and Mahmood, M.},
  title     = {Generative Adversarial Networks-based False Data Injection:
               A Concern for Data Integrity of Power Networks},
  journal   = {IEEE Transactions on Power Systems},
  year      = {2025},
  doi       = {10.1109/TPWRS.2025.3575431}
}

@article{li2024advdiffusionfdia,
  author    = {Li, Kongyijie and Li, Fengyong and Wang, Baonan and
               Shan, Meijing},
  title     = {False data injection attack sample generation using an
               adversarial attention-diffusion model in smart grids},
  journal   = {AIMS Energy},
  volume    = {12},
  number    = {6},
  pages     = {1271--1293},
  year      = {2024},
  month     = jan,
  doi       = {10.3934/energy.2024058}
}

@article{zhao2025piexavae,
  author    = {Zhao, Siliang and Luo, Wuman and Shu, Qin and Xu, Fangwei},
  title     = {Controllable Blind {AC} {FDIA} Via Physics-Informed
               Extrapolative {AVAE}},
  journal   = {Sensors},
  volume    = {25},
  number    = {3},
  pages     = {943},
  year      = {2025},
  doi       = {10.3390/s25030943}
}

@article{jing2025physconstrainedgan,
  author    = {Zhang, Jing and Wang, Qi and Hu, Jianxiong and Wu, Shutan
               and Ye, Yujian and Tang, Yujian},
  title     = {Prior-Knowledge-Free Data Tampering on Power System State
               Estimation: A Physics-Constrained Generative Adversarial
               Framework},
  journal   = {IEEE Internet of Things Journal},
  year      = {2025},
  month     = jan,
  doi       = {10.1109/JIOT.2025.3631113}
}

@article{zhang2024piddpm,
  author    = {Zhang, Shaorong and Cheng, Yuanbin and Yu, Nanpeng},
  title     = {Generating Synthetic Net Load Data with
               Physics-Informed Diffusion Model},
  journal   = {IEEE Transactions on Smart Grid},
  year      = {2024},
  doi       = {10.1109/TSG.2024.3371820}
}

@inproceedings{liu2008iforest,
  author    = {Liu, Fei Tony and Ting, Kai Ming and Zhou, Zhi-Hua},
  title     = {Isolation Forest},
  booktitle = {Proceedings of the 8th IEEE International Conference on
               Data Mining (ICDM)},
  pages     = {413--422},
  year      = {2008},
  doi       = {10.1109/ICDM.2008.17}
}

@inproceedings{sundararajan2017ig,
  author    = {Sundararajan, Mukund and Taly, Ankur and Yan, Qiqi},
  title     = {Axiomatic Attribution for Deep Networks},
  booktitle = {Proceedings of the 34th International Conference on
               Machine Learning (ICML)},
  pages     = {3319--3328},
  year      = {2017}
}

@article{xia2024gcat,
  author    = {Xia, Wei and He, Deming and Yu, Lisha},
  title     = {Locational Detection of False Data Injection Attacks in
               Smart Grids: A Graph Convolutional Attention Network
               Approach},
  journal   = {IEEE Internet of Things Journal},
  volume    = {11},
  number    = {6},
  pages     = {10399--10413},
  year      = {2024},
  month     = mar,
  doi       = {10.1109/JIOT.2023.3323565}
}

@article{baul2023xtm,
  author    = {Baul, Avijit and Sarker, Goutam Chandra and
               Sadhu, Pradip Kumar and Yanambaka, Venkata P. and
               Abdelgawad, Ahmed},
  title     = {{XTM}: A Novel Transformer and {LSTM}-Based Model for
               Detection and Localization of Formally Verified {FDI}
               Attack in Smart Grid},
  journal   = {Electronics},
  volume    = {12},
  number    = {4},
  pages     = {797},
  year      = {2023},
  doi       = {10.3390/electronics12040797}
}

@inproceedings{li2020conaml,
  author    = {Li, Jiangnan and Yang, Yingyuan and
               Tomsovic, Kevin and Sun, Jinyuan Stella and Qi, Hairong},
  title     = {{ConAML}: Constrained Adversarial Machine Learning for
               Cyber-Physical Systems},
  booktitle = {Proceedings of the ACM Asia Conference on Computer and
               Communications Security (ASIA CCS)},
  pages     = {52--62},
  year      = {2021},
  doi       = {10.1145/3433210.3437513}
}

@article{mchugh2012interrater,
  title={Interrater reliability: the kappa statistic},
  author={McHugh, Mary L},
  journal={Biochemia medica},
  volume={22},
  number={3},
  pages={276--282},
  year={2012},
  publisher={Hrvatsko dru{\v{s}}tvo za medicinsku biokemiju i laboratorijsku medicinu}
}

@misc{MITRE,
author = {{MITRE}},
  title = {Techniques - Enterprise | MITRE ATT\&CK®},
  url = "https://attack.mitre.org/techniques/enterprise/",
month = {April },
year = {2026},
  note = "[Online; accessed 2026-04-21]"
}

@ARTICLE{FDIA2016,
  author={Deng, Ruilong and Xiao, Gaoxi and Lu, Rongxing and Liang, Hao and Vasilakos, Athanasios V.},
  journal={IEEE Transactions on Industrial Informatics}, 
  title={False Data Injection on State Estimation in Power Systems—Attacks, Impacts, and Defense: A Survey}, 
  year={2017},
  volume={13},
  number={2},
  pages={411-423},
  doi={10.1109/TII.2016.2614396}}

\end{document}